\documentclass[preprint]{llncs}
\usepackage[utf8]{inputenc}
\usepackage{amssymb}
\usepackage{extarrows}
\usepackage{braket}
\usepackage{xcolor}
\usepackage{soul}
\usepackage{amsmath}
\usepackage{IEEEtrantools}

\usepackage{caption}
\usepackage{mdframed}
\usepackage{dirtytalk}
\usepackage{algorithm}
\usepackage{algorithmic}
\usepackage{graphicx}
\usepackage{float}
\usepackage{authblk}
\usepackage{enumitem}
\usepackage{multicol}
\usepackage{multirow}
\usepackage[export]{adjustbox}
\usepackage{diagbox}
\usepackage{slashbox}
\usepackage{caption} 
\usepackage{sidecap}
\usepackage{hyperref}

\input{Qcircuit}

\newcommand{\ES}{\mathcal{S}}
\newcommand{\E}{\mathcal{E}}
\newcommand{\T}{\mathcal{T}}
\newcommand{\V}{\mathcal{V}}

\newcommand{\F}{\mathcal{F}}
\newcommand{\C}{\mathcal{C}}
\newcommand{\A}{\mathcal{A}}
\newcommand{\Be}{\mathcal{B}_{\epsilon}}
\newcommand{\M}{\mathcal{M}}
\newcommand{\K}{\mathcal{K}}
\newcommand{\Y}{\mathcal{Y}}
\newcommand{\X}{\mathcal{X}}
\newcommand{\R}{\mathcal{R}}
\newcommand{\PRIM}{\mathcal{P}}
\newcommand{\Ora}{\mathcal{O}}
\newcommand{\eO}{\mathcal{O}^{\mathcal{E}}}
\newcommand{\reO}{R \mathcal{O}^{\mathcal{E}}}
\newcommand{\vO}{\mathcal{O}^{\mathcal{V}}}

\newcommand{\GCM}[2]{\mathcal{G}^{#1}_{#2}}

\newcommand{\nul}{\mathsf{null}}
\newcommand{\qEx}{\mathsf{qEx}}
\newcommand{\qSel}{\mathsf{qSel}}
\newcommand{\qUni}{\mathsf{qUni}}
\newcommand{\negl}{negl}
\newcommand{\nonnegl}{non\text{-}\negl}
\newcommand{\Hil}{\mathcal{H}}

\newcommand{\HilD}{\mathcal{H}^D}

\newcommand{\HilR}{\mathcal{H}^{\mathcal{R}}}
\newcommand{\Ue}{\mathrm{U_{\mathcal{E}}}}
\newcommand{\U}{\mathrm{U}}
\newcommand{\notmu}{\not\in_{\mu}}

\newcommand{\euf}{1\text{-}qGEU}
\newcommand{\eufm}{$\mu$\text{-}qGEU}

\newcommand{\suf}{1\text{-}qGSU}
\newcommand{\sufm}{$\mu$\text{-}qGSU}

\newcommand{\uuf}{qGUU}
\newcommand{\bz}{BZ}
\newcommand{\bu}{BU}
\newcommand{\qprf}{qPRF}
\newcommand{\prf}{PRF}
\newcommand{\pru}{PRU}
\newcommand{\uu}{UU}
\newcommand{\mbraket}[2]{\bra{#1}#2\rangle}

\newenvironment{boxfig}[2]{%
     \begin{figure} [th!]
     \newcommand{\FigCaption}{#1}
     \newcommand{\FigLabel}{#2}
     \begin{center}
       \begin{small}
         \begin{tabular}{@{}|@{~~}l@{~~}|@{}}
           \hline
           \rule[-1.5ex]{0pt}{1ex}\begin{minipage}[b]{.95\linewidth}
            \vspace{1ex}
             \smallskip
             }{%
           \end{minipage}\\
           \hline
         \end{tabular}
       \end{small}

     \end{center}
      \caption{\FigCaption}
     \label{\FigLabel}
   \end{figure}
}

\begin{document}
\title{A Unified Framework For Quantum Unforgeability}
\author{Mina Doosti\inst{1} \and Mahshid Delavar\inst{1} \and 
  Elham Kashefi\inst{1,2} \and Myrto Arapinis\inst{1}}

\institute{School of Informatics, University of Edinburgh,\\
10 Crichton Street, Edinburgh EH8 9AB, UK\\
\and
Departement Informatique et Reseaux, CNRS, Sorbonne Universit\'{e},\\
4 Place Jussieu 75252 Paris CEDEX 05, France}


\maketitle     



\begin{abstract}
In this paper, we continue the line of work initiated by Boneh and Zhandry at CRYPTO 2013 and EUROCRYPT 2013 in which they formally define the notion of unforgeability against quantum adversaries. We develop a general and parameterised quantum game-based security model unifying unforgeability both for classical and quantum constructions allowing us for the first time to present a complete quantum cryptanalysis framework for unforgeability. In particular, we prove how our definitions subsume previous ones while considering more fine-grained adversarial models, capturing the full spectrum of superposition attacks. The subtlety here resides in the characterisation of a forgery.
We show that the strongest level of unforgeability in our framework, namely existential unforgeability, can only be achieved if only orthogonal to previously queried messages are considered to be forgeries.
We further show that deterministic constructions can only achieve the weaker notion of unforgeability, that is selective unforgeability, against such adversaries, but that selective unforgeability breaks if more general quantum adversaries (capable of general superposition attacks) are considered. On the other hand, we show that a PRF is sufficient for constructing a selective unforgeable classical primitive against full quantum adversaries. 
\end{abstract}

\section{Introduction}
Recent advances in quantum technologies threaten the security of many widely-deployed cryptographic primitives. This calls for quantum-secure cryptographic schemes. Usually, two main security models are considered when analysing the security of cryptographic primitives against quantum adversaries: the \emph{standard security} model, often also termed post-quantum security, where the adversary only has classical access to the primitive but can locally perform quantum computations; or the \emph{quantum security} model where the adversary has further quantum access to the primitive, \emph{i.e.} they can issue quantum queries.
In the quantum setting, and more specifically in the quantum security model, the quantum nature of interaction with the primitives, enables a broader range of attack scenarios, making the task of transposing security definitions to the quantum setting highly non-trivial and subtle~\cite{boneh2013quantum,boneh2013secure,kaplan2016breaking,gagliardoni2016semantic,eurocrypt-2020-30239}.
One of the key elements of the quantum security model is the fact that the adversary can query the oracle with quantum states in superposition. Superposition queries are more likely to lead to non-trivial attacks~\cite{kaplan2015quantum,santoli2016using} that are not possible in the classical regime. Another important aspect is that having access to the input-output pairs of the oracle in the form of quantum states enables the adversary to run quantum algorithms and take advantage of quantum speedup. Of course, a possible countermeasure against \emph{superposition attacks} is to forbid any kind of quantum access to the oracle through measurements. However, in such a setting the security relies on the physical implementation of the measurement tool which itself could be potentially exploited by a quantum adversary. Thus, and as it has previously been advocated in~\cite{boneh2013quantum,boneh2013secure,kaplan2016breaking,eurocrypt-2020-30239}, providing security guarantees in the quantum security model is crucial. In this paper, we pursue the line of work initiated by Boneh and Zhandry in~\cite{boneh2013quantum,boneh2013secure}, as well as Alagic \emph{et al.} in \cite{eurocrypt-2020-30239} on formalizing the notion of unforgeability in the quantum security model. This notion is the security property desired for many primitives such as Message Authentication Codes, Digital Signatures, or Physical Unclonable Functions. Informally, unforgeability ensures that the adversary cannot produce valid input-output pairs of the oracle without access to the full description of its circuit. These previous definitions, as we will see, do not, however, capture the full spectrum of possible superposition attacks. Unforgeability is also a key security property for quantum primitives, such as Quantum Physical Unclonable Functions (qPUF) and Quantum Money; however, previous definitions~\cite{boneh2013quantum,boneh2013secure,eurocrypt-2020-30239} again do not apply to such quantum primitives. 

\subsection{Different levels of Classical and Quantum Unforgeability}
Goldwasser \emph{et al.}~\cite{goldwasser1988digital} define different notions of unforgeability for digital signatures. They consider various types of attacks including: \emph{chosen message attacks (cma)} where the adversary is allowed access to the signing oracle on a list of messages chosen by the adversary. They define \emph{existential forgery} as the attack where the adversary can forge a valid signature for at least one new message; and the notion of \emph{selective forgery} as an attack where the adversary can forge a valid signature with non-negligible probability for a particular message chosen by the adversary prior to accessing the signing oracle.

An \emph{et al.} \cite{an2002security} define a slightly stronger notion of unforgeability called \emph{strong unforgeability} that requires the adversary not only to be unable to generate a valid signature on a``new" message but also to be unable to generate even a valid ``new" signature on an already signed message. \emph{Strong Existential Unforgeability (SEUf)}, also called \emph{strong unforgeability}, has formally been defined in \cite{boneh2006strongly} by Boneh \emph{et al.}

Bellare \emph{et al.}~\cite{bellare1995xor} define the notion of Strong Existential Unforgeability under chosen message and chosen verification queries attack (SEUF-cmva) for message authentication codes (MACs). In both of these attack models, the adversary is allowed a chosen message oracle access, as defined for digital signatures in~\cite{goldwasser1988digital}. Although in the later attack model for message authentication codes, the experiment also allows verifying queries through oracle access. This model is justified for MACs as unlike digital signatures, where the verification algorithm is public, the adversary cannot run the verification algorithm on their own. \emph{(Weak) Existential Unforgeability (EUf) under chosen message attacks} is a natural definition for MACs defined by Bellare \emph{et al.}~\cite{bellare2000security} and comes by extending the one for digital signatures~\cite{goldwasser1988digital}.

Moreover, Dodis \emph{et al.}~\cite{dodis2012message} define the notion of \emph{selective unforgeability under adaptive chosen message and chosen verification queries (SelUF-cmva)}.

A yet weaker notion called \emph{universal unforgeability} requires the adversary to produce a fresh tag for a uniformly random message given as challenge to the adversary~\cite{alwen2014key}. This notion again can be considered against both attack models: chosen message and chosen verification queries attack (UniUF-cmva) and chosen message attack (UniUF-cma).

Table~\ref{table:cunf} summarizes all these different classical notions of unforgeability.

\begin{table}[h]
\centering
\captionsetup{font=small}
\resizebox{0.6\textwidth}{!}{
\begin{tabular}{|c|c|c|}
\hline
\backslashbox{Def. level}{Attack Model} & cmva & cma \\ \hline
SEUf (strong) & - & \cite{an2002security,boneh2006strongly,bellare1995xor} \\ \hline
EUf (weak) & \cite{dodis2012message,bellare2004power} & \cite{boneh2006strongly,bellare2000security} \\ \hline
SelUf (selective) & - & \cite{dodis2012message} \\ \hline
UniUf (universal) & \cite{alwen2014key} & \cite{alwen2014key} \\ \hline 
\end{tabular}}
\caption{Classical unforgeability definitions from strongest to weakest. \textit{cmva} - adaptive chosen message queries and limited access to the verification oracle. \textit{cma} - (adaptive) chosen message attacks. Cases marked with ``-", no definition has been proposed yet to the best of our knowledge.}
\label{table:cunf}
\end{table}

In the quantum regime, the definition of unforgeability defined by Boneh and Zhandry~\cite{boneh2013quantum,boneh2013secure} (denoted by \bz), is described as a quantum analogue of \emph{strong existential unforgeability} and it is in the chosen message attack (cma) model. The definition of \emph{Blind unforgeability (\bu)} by Alagic \emph{et al.}~\cite{eurocrypt-2020-30239} has been defined as \emph{(weak) quantum existential unforgeability} but they have also presented the extension of the definition to strong existential unforgeability. In this paper, we present a unified and parameterised definition that extends to different levels of unforgeability. Our \emph{Quantum Generalised Existential Unforgeability (\eufm)} has been defined as a quantum analogue of (weak) existential unforgeability, although we will show that it can be extended to capture the strong case as well. Further we investigate the quantum analogue of \emph{selective} and \emph{universal} unforgeability, namely \emph{Quantum Generalised Selective Unforgeability (\sufm)} and \emph{Quantum Generalised Universal Unforgeability (\uuf)}. Our formal definitions have been defined in the \emph{cma} attack model similar to previous ones although the structure of our game easily allows for weaker attack models, namely \emph{random message attack (rma)} in which the adversary instead of choosing arbitary queries, is given a set of randomly selected quantum states from a distribution, and their respective query outputs. This attack model is most relevant for quantum primitives, with potential interest for some classical primitives studied in the quantum security model as well. Nevertheless, we will not discuss this attack model in details in the current paper. Finally, we also study an adaptive attack model which is specific to universal unforgeability and allows the adversary to continue querying the oracle after receiving a randomly picked message as a challenge. Table~\ref{table:qunf} shows a summary of different levels of quantum unforgeability introduced in previous works, as well as the current paper.

\begin{table}
\centering
\captionsetup{font=small}
\resizebox{\textwidth}{!}{
\begin{tabular}{|c|c|c|c|c|}
\hline
\backslashbox{Def. level}{Attack Model} & cmva & cma & rma & aua \\ \hline
qSEUf (strong) & - & \bz~\cite{boneh2013quantum,boneh2013secure}, \bu*~\cite{eurocrypt-2020-30239}, \emph{this}(\eufm*) & NA & NA\\ \hline
qEUf (weak) & - & \bu~\cite{eurocrypt-2020-30239}, \emph{this}(\eufm) & \emph{this}(\eufm) & NA\\ \hline
qSelUf (selective) & - & \emph{this}(\sufm) & \emph{this}(\sufm) & NA \\ \hline
qUniUf (universal) & - & \emph{this}(\uuf) & \emph{this}(\uuf) & \emph{this}(\uuf) \\ \hline
\end{tabular}}
\caption{Quantum unforgeability definitions from strongest to weakest. The definitions introduced in this paper are prefixed with ``\emph{this}". The following attack models are considered \emph{cmva} - adaptive message queries and also limited access to the verification oracle. \emph{cma} - (adaptive) chosen message attacks. \emph{rma} - random (and unknown) message attacks. \emph{aua} is a specific adaptive attack model only applicable for universal unforgeability. Some attack models are not applicable for some of the definitions, denoted by (NA). Cases marked with ``-", no definition has been proposed yet to the best of our knowledge.}
\label{table:qunf}
\end{table}

\subsection{Our Contributions}
We propose a general and unified definition of quantum unforgeability for both classical and quantum cryptographic primitives. Our definition captures any quantum adversary, covering the full spectrum of superposition attacks.
We present our definitions in the quantum-game based framework in the spirit of~\cite{boneh2013secure,gagliardoni2017quantum,alagic2018unforgeable}. Our framework generalises the notion of unforgeability in three aspects. First, by generalising the message space to both the classical message space and Hilbert spaces, and allowing a wider range of quantum oracle access types, we unify the notion of quantum unforgeability for both quantum and classical primitives.

Second, our framework captures different levels of unforgeability as quantum analogues of the unforgeability notions studied in the classical setting. These levels correspond to different attacker capabilities and have different practical applications. More precisely, previous definitions of quantum unforgeability only capture \emph{strong} and \emph{weak existential unforgeability}, while our framework further captures the notions of \emph{selective} and \emph{universal unforgeability} for the first time. We formally show the hierarchy between these definitions through our framework. 

Finally, our framework precisely captures the quantum capabilities of the adversary in terms of overlap between the challenge and the queried states in the learning phase. This formalizes the full spectrum of unforgeability from classical to fully quantum, revealing new non-trivial attacks. Our parameterised definitions of $\mu$-existential and $\mu$-selective unforgeability allow the adversary to forge a ``new" $\mu$-distinguishable challenge. The notion of $\mu$-distinguishability captures the overlap between the challenge and the learning phase and allows characterising ``new" challenges in a fine-grained manner. This contrasts with previous definitions which characterise ``new" challenges, respectively, through counting the queries like Boneh-Zhandry~\cite{boneh2013quantum,boneh2013secure}. This approach is too weak as previously pointed by Alagic \emph{et al.}~\cite{eurocrypt-2020-30239}, and does not fully explore the advantage that a quantum adversary can gain through quantum queries and the fact that some quantum queries are being consumed during the attack. Moreover, we formally show that the definition of~\cite{eurocrypt-2020-30239} is a special instance of our definition. We then explore the applicability and relevance of our definitions through several novel possibility and impossibility results. Here we give a summary of our key findings.

\subsubsection{Generalised Existential Unforgeability (\eufm):} We show that this notion of unforgeability can only be achieved in the most restricted case ($\mu = 1$) where the adversary is not allowed any overlap between their queries to the oracle (during the learning phase) and the target forgery message. For any other value of $\mu$, we show the existence of a general superposition attack and hence that no quantum or classical primitive can satisfy existential unforgeability. Nevertheless, as we show the equivalence with \bu\ for $\mu=1$, we inherit the positive results from~\cite{eurocrypt-2020-30239} for this case.

\subsubsection{Generalised Selective Unforgeability (\sufm):} This is a weaker unforgeability notion, where the adversary needs to commit their selected messages before querying the oracle in the learning phase. Here our results show a non-intuitive impossibility as well as a separation between randomised and non-randomised constructions. Our definition carefully discards the probability of trivial attacks, to only capture effective adversaries. First, we prove that no classical or quantum primitive with a deterministic evaluation algorithm satisfies this notion of unforgeability. To establish our impossibility result, we show an attack based on the Universal Quantum Emulator Algorithm~\cite{marvian2016universal}. 
This type of attack was first studied in the context of quantum physical unclonable functions~\cite{arapinis2019quantum}. Here we show that similar attacks apply to some levels of unforgeability for classical primitives too. Concretely, our no-go result implies that deterministic Message Authentication Codes constructions such as HMAC, NMAC, \emph{etc.} cannot satisfy \eufm~nor \sufm~except for quantum adversaries restricted to orthogonal challenges (case where $\mu = 1$). Hence these classical primitives are always vulnerable against more powerful quantum adversaries, or in other words when these distinguishability conditions cannot be efficiently checked and implemented in practice. On the other hand, we show that Pseudorandom Functions (PRFs) are sufficient for constructing a quantum selective unforgeable classical primitive against full quantum adversaries (for all reasonable degrees of $\mu$) by proposing a randomised construction. Similarly, we present a randomised quantum primitive that can satisfy the same unforgeability level relying on the assumption of Pseudorandom Unitaries (PRU). For our quantum construction, we also characterize the quantum randomised oracle and propose a construction in the circuit model.

\subsubsection{Generalised Universal Unforgeability (\uuf):} Here the notion of unforgeability is further weakened requiring the adversary to forge the response to a message picked uniformly at random by the challenger. We show general positive results for both quantum and classical primitives \emph{wrt} this notion, provided their evaluation algorithm is a quantum secure PRF or a PRU. We note that even though this notions is weaker than the previous ones, it is much easier to achieve in practice and has applications in many scenarios such as quantum identification protocols~\cite{doosti2020client}.

\section{Preliminaries}\label{sec:prelim}
In this section, we discuss the previous definitions for quantum unforgeability, as well as some of the main concepts and definitions that we rely upon in the paper. 



\subsection{Quantum accessible oracles for classical primitives}\label{sec:prelim-oracles}
A quantum oracle is a unitary transformation $\bold{\Ora}$ over a $D$-dimensional Hilbert space that can be queried with quantum queries. The quantum oracle can grant quantum access to the evaluation transformation of a classical or quantum primitive. For classical primitives we follow the standard definition of quantum oracle\cite{boneh2013quantum,boneh2013secure,gagliardoni2016semantic,gagliardoni2020make,chevalier2020security}


In the standard quantum-query model, the adversary $\A$ has black-box access to a reversible version of $f$, which is a classical-polynomial-time computable deterministic or randomised function of the evaluation $\E$, through an oracle $\reO_f$ which is a unitary transformation. The evaluation oracle can be represented as:
\begin{equation}
    \reO_f: \sum_{m,y} \alpha_{m,y} \ket{r}_{\Ora}\ket{m,y} \rightarrow \sum_{m,y} \alpha_{m,y}\ket{r}_{\Ora}\ket{m, y \oplus f(m;r)}
\end{equation}
This is also referred to as \emph{Standard Oracle}. Here $m$ is the message and $y$ is the ancillary system required for unitarity. In general the standard oracle can also capture randomised evaluations with a randomness $r$ picked from $\R \subseteq \{0,1\}^l$ as the randomness space, although in this case the oracle may not be a unitary transformation. The unitary representation of the standard oracle has been introduced in several works such as~\cite{gagliardoni2016semantic,gagliardoni2020make,chevalier2020security} with slightly different approaches that lead to equivalent adversary's state, which is totally mixed with respect to the randomness subspace. Although in this work, to emphasise that the adversary cannot gain access to the internal randomness register of the oracle directly and avoid some potential artificial entanglement attacks, we opt for the approach of~\cite{gagliardoni2020make} and consider the randomness as an internal state of the oracle which is re-initiated for each query with a new classical value $r$. This choice is also due to the fact that the oracle needs to output the randomness register as a separable state, otherwise an unwanted entanglement will be created between the adversary's output state and the internal register of the oracle, as also mentioned in~\cite{gagliardoni2020make}. Moreover if the primitive requires that the randomness is returned to the adversary for each query (as a classical bit-string or a function of $r$), it can be recorded in the adversary's auxiliary state $y$ that can be extended to also capture the randomness space. An example of such construction will be introduced later in the paper. Finally, we specify that for deterministic primitives (denoted by $\eO_f$), the structure is similar except that the randomness register is not used.

\subsection{Quantum oracle for quantum primitives}
The evaluation between these states of a quantum primitive can be directly defined as a unitary transformation. Hence the deterministic oracle can be modeled as follows:
\begin{equation}\label{qoracle-quantum-det}
    \eO_U: \sum_{i} \alpha_{i} \ket{m_i} \overset{U_{\mathcal{E}}}{\rightarrow} \sum_{i} \beta_{i}\ket{m_i}
\end{equation}
where $\{\ket{m_i}\}$ are a basis (not necessary computational basis) for $\HilD$ that the unitary operates upon. We note that the quantum primitives can perform an arbitrary rotation of the bases. The analogue of this type of oracles for classical primitives, are \emph{type-2} oracles (also called \emph{minimal oracles})\cite{gagliardoni2016semantic,gagliardoni2020make}. A randomised quantum primitive can also be defined similar to the classical case. Here we give an abstract notation of a general randomised quantum primitive, but we further clarify the realisation of such oracles in the upcoming sections. We denote a general randomised unitary oracle for quantum primitives as follows:
\begin{equation}\label{qoracle-quantum-rand}
\reO_U: \sum_{i} \alpha_{i} \ket{r}_{\Ora} \ket{m_i} \overset{U_{\mathcal{E}}}{\rightarrow} \sum_{i} \beta_{i}(r)\ket{r}_{\Ora}\ket{m_i}
\end{equation}
Hence a $\reO$ is a unitary over the joint space of the oracle's randomness register and the main input state, which consist of a family of smaller unitaries parameterised by a random internal parameter $r$.

\subsection{Formal definitions of \bu\ and \bz}
The definition of existential unforgeability under quantum chosen-message attacks (EUF-qCMA) for digital signatures has been presented in \cite{boneh2013secure,boneh2013quantum} by Boneh and Zhandry and is defined as follows.

\begin{definition}\label{def:bz}[\bz (EUF-qCMA) \cite{boneh2013quantum}]
A system S (Sign/Mac), is existentially unforgeable under a quantum chosen message attack (EUF-qCMA) if no adversary after issuing $q$ quantum chosen message queries, can generate $q+1$ valid classical message-tag pairs with non-negligible probability in the security parameter.
\end{definition}

Another definition of unforgeability against quantum adversaries called \emph{blind unforgeability} was proposed in~\cite{eurocrypt-2020-30239}. This more recent definition aims to capture some attacks that are not captured by \bz. This notion defines an algorithm to be forgeable if there exists an adversary who can use access to a ``partially blinded" oracle to validate responses of the messages that are in the blinded region and hence only respond to the queries that are not in this region.
A blinded operation for a function $f: X \rightarrow Y$ and a subset of messages $B \subseteq X$ is defined as:
\begin{equation}
    Bf(x)= 
    \begin{cases}
    \perp, & \text{if } x \in B\\
    f(x), & \text{otherwise}
    \end{cases}
\end{equation}
Where in particular for the definition of unforgeability, the elements of X are placed in B independently at random with a particular probability $\epsilon$, denoted by $B_{\epsilon}$. Then the security game of unforgeability has been defined as follows with the adversary having access to the blinded oracle.
\begin{definition}\label{def:bu}[\cite{eurocrypt-2020-30239}(Def.$4 \& 5$)]
Let $\Pi = (KeyGen,Mac,Ver)$ be a MAC with message set $X$. Let $\A$ be an algorithm, and $\epsilon: \mathbb{N} \rightarrow \mathbb{R}_{\geq 0}$ an efficiently computable function. The blind forgery experiment $BlindForge_{\A,\Pi}(n,\epsilon)$ proceeds as follows:
\begin{enumerate}
    \item Generate key: $k \leftarrow KeyGen(1^n)$
    \item Generate blinding: select $B_{\epsilon}\subseteq X$ by placing each $m$ into $B_{\epsilon}$ independently with probability $\epsilon(n)$.
    \item Produce forgery: $(m,t) \leftarrow \A^{B_{\epsilon}{MAC}_k}(1^n)$.
    \item Outcome: output 1 if $Ver_k(m,t) = acc$ and $m \in B_{\epsilon}$ ; otherwise output 0.
\end{enumerate}
From this game blind-unforgeability is defined as follows.\\
A MAC scheme $\Pi$ is blind-unforgeable (BU) if for every polynomial-time uniform adversary $(\A,\epsilon)$
\[
    Pr[BlindForge_{\A,\Pi}(n,\epsilon(n)) = 1] \leq \negl(n).
\]
and the probability is taken over the choice of key, the choice of blinding set, and any internal randomness of the adversary.
\end{definition}
Thus, in this definition, a forgery happens if the adversary can produce a valid tag for a message within the blinded region. We refer to this definition of unforgeability as \bu. This definition imposes that the challenge be orthogonal to the previously queried messages.

We also recall the following theorem from~\cite{eurocrypt-2020-30239} which we will use later in the paper:
\begin{theorem}\label{th:bu}[from~\cite{eurocrypt-2020-30239}]
Let $\A$ be a QPT such that $supp(\A)\cap R = \emptyset$\footnote{Here $supp(\A)$ denotes the support of $\A$ that is defined as follows. Let $\A$ have oracle access to a classical function $f:\{0,1\}^n \rightarrow \{0,1\}^m$. Let $\ket{\psi_i}$ be the state of the the query $i$ or equivalently the intermediate state after applying $U_i$ in the sequence of $\Ora U_q\Ora\dots U_1$ on an initial state $\ket{0}_{XYZ}$ where $X$ denotes the input registers. Then $supp(\A)$ is defined to be the set of input strings $x$ such that there exists a function $f$ with the respective oracle such that $\mbraket{x}{\psi_i}_X \neq 0$ for at least one of the queries.} for some $R \neq \emptyset$. Let $\mathtt{MAC}$ be a MAC, and suppose $A^{\mathtt{MAC}_k}(1^n)$ outputs a valid pair $(m,\mathtt{Mac}_k(m))$ with $m \in R$ with non-negligible probability. Then $\mathtt{MAC}$ is not BU-secure.
\end{theorem}

\subsection{Distinguishability of quantum states}\label{sec:prelim-qtest}
An important difference between quantum and classical bits is the impossibility of creating perfect copies of general unknown quantum states, known as the \emph{no-cloning theorem} \cite{wootters1982single}. This is an important limitation imposed by quantum mechanics which is particularly relevant for cryptography. A variation of the same feature states that it is impossible to obtain the exact classical description of a quantum state by having a single copy of it. Therefore, there exists a bound on how well one can derive the classical description of quantum states depending on their dimension and the number of available copies. Hence, distinguishing between unknown quantum states can be achieved only probabilistically. A useful and relevant notion of quantum distance that we exploit in this paper is \emph{fidelity}. Generally the fidelity of mixed states $\rho$ and $\sigma$ is defined by the Uhlmann fidelity:
\begin{equation}
    F(\rho,\sigma)=[Tr(\sqrt{{\sqrt{\rho}}\sigma {\sqrt{\rho}}})]^{2}.    
\end{equation}
Which gives $F(\ket{\psi}, \ket{\phi}) = |\mbraket{\psi}{\phi}|^2$ for two pure quantum states $\ket{\psi}$ and $\ket{\phi}$.
Distinguishability and indistinguishability are well known concepts in quantum information and have been stated with different quantum distance measures such as trace distance or fidelity. Here we use the fidelity-based notion of $\mu$-distinguishability defined as follows: 
\begin{definition}[$\mu$-distinguishability] \label{def:dist}
Let $F(\cdot,\cdot)$ denote the fidelity, and $0 \leq \mu \leq 1$ the distinguishability threshold respectively. We say two quantum states $\rho$ and $\sigma$ are $\mu$-distinguishable if $0\le F(\rho,\sigma) \leq 1-\mu$.
\end{definition}
Note that two quantum states, $\rho$ and $\sigma$, are \emph{completely distinguishable} or 1-distinguishable ($\mu=1$), if $F(\rho,\sigma)= 0$.

\subsection{Verifying quantum states}\label{app:tools-quantum}
Due to the impossibility of perfectly distinguishing between all quantum states according to the above definition, checking equality of two completely unknown states is a non-trivial task. This is one major difference between classical bits and qubits. Nevertheless, a probabilistic comparison of unknown quantum states can be achieved through the simple quantum SWAP test algorithm~\cite{buhrman2001quantum}. The SWAP test and its generalisation to multiple copies introduced recently in~\cite{chabaud2018optimal}. We also give an abstract definition for a general quantum test algorithm and define its necessary conditions.

\begin{definition}[Quantum Testing Algorithm]\label{def:test} Let $\rho^{\otimes \kappa_1}$ and $\sigma^{\otimes \kappa_2}$ be $\kappa_1$ and $\kappa_2$ copies of two quantum states $\rho$ and $\sigma$, respectively. A Quantum Testing algorithm $\T$ is a quantum algorithm that takes as input the tuple ($\rho^{\otimes \kappa_1}$,$\sigma^{\otimes \kappa_2}$) and accepts $\rho$ and $\sigma$ as equal (outputs 1) with the following probability
\begin{equation*}
\mathrm{Pr}[1 \leftarrow \T(\rho^{\otimes \kappa_1}, \sigma^{\otimes \kappa_2})] = 1 - \mathrm{Pr}[0 \leftarrow \T(\rho^{\otimes \kappa_1}, \sigma^{\otimes \kappa_2})] = f(\kappa_1,\kappa_2, F(\rho, \sigma))
\end{equation*}
where $F(\rho, \sigma)$ is the fidelity of the two states and $f(\kappa_1,\kappa_2, F(\rho, \sigma))$ satisfies the following limits:
\begin{equation*}
 \begin{cases}
    \lim_{F(\rho, \sigma) \rightarrow 1}f(\kappa_1,\kappa_2, F(\rho, \sigma)) = 1  & \forall\:(\kappa_1,\kappa_2)\\
    \lim_{\kappa_1,\kappa_2 \rightarrow \infty}f(\kappa_1,\kappa_2, F(\rho, \sigma)) = F(\rho, \sigma)\\
    \lim_{F(\rho, \sigma) \rightarrow 0}f(\kappa_1,\kappa_2, F(\rho, \sigma)) = Err(\kappa_1, \kappa_2)
  \end{cases} 
\end{equation*}
with $Err(\kappa_1,\kappa_2)$ characterising the error of the test algorithm.
\end{definition}

\subsection{Quantum Emulation Algorithm}\label{sec:prelim-qe}
In this section, we describe the Quantum Emulation (QE) algorithm presented in \cite{marvian2016universal} as a quantum process learning tool and used in~\cite{arapinis2019quantum} for the first time as an attack algorithm against a quantum primitive, namely quantum physical unclonable functions. The main purpose of quantum emulation is to mimic the action of an unknown unitary transformation on an unknown input quantum state by having some of the input-output samples of the unitary. An emulator is not trying to completely recreate the transformation or simulate the same dynamics. Instead, it outputs the action of the transformation on a quantum state. This task is done by construction of some controlled-reflection gates that first project the input state in the subspace of the input samples while encoding the information in ancillary systems. Then by using controlled-reflection around the output state the components of the state are retrieved while the unitary is applied to the state.

We are interested in the fidelity of the output state $\ket{\psi_{QE}}$ of the algorithm and the intended output $\U\ket{\psi}$ to estimate the success. Hence we recall two main theorems from~\cite{marvian2016universal} and~\cite{arapinis2019quantum} which we use in the proof of Theorem~\ref{th:sel-qCM}.

The first theorem states that the final fidelity is lower-bounded by the square root of the success probability of the projection into the input subspace in the first step.
\begin{theorem}\label{th:qe-fidel}\textbf{\cite{marvian2016universal}}
Let $\E_{\U}$ be the quantum channel that describes the overall effect of the algorithm presented above. Then for any input state $\rho$, the Uhlmann fidelity of $\E_{\U}(\rho)$ and the desired state $\U\rho\U^{\dagger}$ satisfies:
\begin{equation}\label{eq:qe-fidel}
    F(\rho_{QE}, \U\rho\U^{\dagger}) \geq F(\E_{\U}(\rho), \U\rho\U^{\dagger}) \geq \sqrt{P_{succ-stage1}}
\end{equation}
where $\rho_{QE} = \ket{\psi_{QE}}\bra{\psi_{QE}}$ is the main output state (tracing out the ancillas) after the first step. $\E_{\U}(\rho)$ is the output of the whole circuit without the post-selection measurement in the second stage and $P_{succ-stage1}$ is the success probability of the first step.
\end{theorem}
Also the success probability of Stage 1 is calculated as follows,
\begin{equation}\label{eq:qe-ps}
P_{succ-stage1} = |\bra{\phi_r}Tr_{anc}(\ket{\chi_f}\bra{\chi_f})\ket{\phi_r}|^2
\end{equation}
Where $\ket{\chi_f}$ is the final overall state of the emulation's algorithm first stage on input state $\ket{\psi}$, and $\ket{\phi_r}$ is the reference state as defined in~\cite{arapinis2019quantum}.
We also recall a simplified version of a theorem in \cite{arapinis2019quantum} as follows:
\begin{theorem}\label{th:qe-fins}{\cite{arapinis2019quantum} (simplified)}
Let $\ket{\chi_1}$ be the final overall state of the circuit after one block. The final state is of the following form:
\begin{equation}\label{eq:qe-fins}
\begin{split}
    \ket{\chi_1} &= \mbraket{\phi_r}{\psi} \ket{\phi_r}\ket{0} + \ket{\psi}\ket{1} - \mbraket{\phi_r}{\psi} \ket{\phi_r}\ket{1} -2\mbraket{\phi_1}{\psi} \ket{\phi_1}\ket{1} \\
    & +2\mbraket{\phi_r}{\psi}\mbraket{\phi_r}{\phi_1} \ket{\phi_1}\ket{1}
\end{split}
\end{equation}
\end{theorem}

Having a precise expression for $\ket{\chi_f}$ from Theorem~\ref{th:qe-fins}, one can calculate $P_{succ-step1}$ of equation~(\ref{eq:qe-ps}) by tracing out the ancillary systems from the density matrix of $\ket{\chi_f}\bra{\chi_f}$.

\subsection{Quantum-secure Pseudorandom Function (\qprf)}
Quantum-secure Pseudorandom Functions are families of functions that look like truly random functions to QPT adversaries. Formally, \qprf s are defined as follows.

\begin{definition}\label{def:qprf}[Quantum-Secure Pseudorandom Functions(\prf): \cite{ji2018pseudorandom}]
Let $\K$, $\X$, $\Y$ be the key space, the domain and range respectively, all implicitly depending on the security parameter $\lambda$. A keyed family of functions $\{PRF_k: \X \rightarrow \Y\}_{k\in \K}$ is a quantum-secure pseudorandom function (\prf) if for any polynomial-time quantum oracle algorithm $\A$, $PRF_k$ with a  random $k \leftarrow \K$ is indistinguishable from a truly random function $f \leftarrow \Y^{\X}$ in the sense that:
\begin{equation}
|\underset{k \leftarrow \K}{Pr}[\A^{PRF_k}(1^{\lambda})=1] -\underset{f \leftarrow \Y^{\X}}{Pr}[\A^{f}(1^{\lambda})=1]| = \negl(\lambda).
\end{equation}
\end{definition}
\medskip{}

\subsection{Quantum Pseudorandomness}\label{sec:prelim-qpseudorand}
Pseudorandomness is a central concept is modern cryptography which has also been extended to the quantum regime. We have defined the notion of quantum-secure Pseudorandom Functions in Section~\ref{sec:prelim}. Here we define its quantum analogue, namely quantum Pseudorandom Unitaries (\pru), as well as another related notion called Unknown Unitary (\uu).

\begin{definition}\label{def:pru}[Pseudorandom Unitary Operators(\pru): \cite{ji2018pseudorandom}]
A family of unitary operators $\{U_k \in \mathcal{U}(\Hil)\}_{k \in \mathcal{K}}$ is a pseudorandom unitary if two conditions hold:
\begin{itemize}
    \item \textbf{Efficient computation}. There is an efficient quantum algorithm $Q$ such that forall $k$ and any state $\ket{\psi} \in S(\Hil), Q(k,\ket{\psi}) = U_k\ket{\psi}$.
    \item \textbf{Pseudorandomness}. $U_k$ with a random key $k$ is computationally indistinguishable from a Haar random unitary operator. More precisely, for any efficient quantum algorithm $\A$ that makes at most polynomially many queries to the oracle:  
\end{itemize}
\begin{equation}
    |\underset{k \leftarrow \K}{Pr}[\A^{U_k}(1^{\lambda})=1] - \underset{U \leftarrow \mu}{Pr}[\A^U(1^{\lambda})=1]| = \negl(\lambda).
\end{equation}
where $\mu$ is the Haar measure on $S(\Hil)$. Note that here we focus on the Pseudorandomness condition of the PRU definition.
\end{definition}

We also mention a relevant notion to \pru, called family of Unknown Unitaries (\uu) defined in~\cite{arapinis2019quantum}, that can also be interpreted as single-shot pseudorandomness.

\begin{definition}[Unknown Unitary Transformation]\label{def:uu} We say a family of unitary transformations $U^u$, over a $D$-dimensional Hilbert space $\HilD$ is called Unknown Unitaries, if for all QPT adversaries $\A$ the probability of estimating the output of $U^u$ on any randomly picked state $\ket{\psi}\in\HilD$ is at most negligibly higher than the probability of estimating the output of a Haar random unitary operator on that state:
\begin{equation}\small
        |\underset{U \leftarrow U^u}{Pr}[F(\A(\ket{\psi}),U\ket{\psi}) \geq \nonnegl(\lambda)] - \underset{U_{\mu} \leftarrow \mu}{Pr}[F(\A(\ket{\psi}),U_{\mu}\ket{\psi}) \geq \nonnegl(\lambda)]| = \negl(\lambda).
\end{equation}
\end{definition}

In the remainder of the paper, we will let $\lambda$ denote the security parameter. A non-negative function $\negl(\lambda)$ is negligible if, for any constant $c$, $\negl(\lambda)\le \frac{1}{\lambda^c}$ for all sufficiently large $\lambda$.

\section{Generalized Quantum Unforgeability}\label{sec:unf-def}
The game-based security framework is a standard model for formally defining security properties of cryptographic primitives such as encryption algorithms, digital signature schemes or physical unclonable functions~\cite{boneh2013secure,gagliardoni2016semantic,alagic2018unforgeable,armknecht2016towards,soukharev2016post}. Classical cryptographic primitives have also widely been studied in a quantum game-based framework, where parties are Quantum Turing Machines (QTM)~\cite{boneh2013secure,gagliardoni2017quantum,alagic2018unforgeable,soukharev2016post}. Inspired by these works, we generalise the quantum game-based framework to define quantum unforgeability. Our definitions unify different levels of unforgeability as well capturing quantum and classical primitives. In this section we mostly focus on classical primitives. In Section~\ref{sec:qGUnf-quantum}, we show how the framework can naturally cater for quantum primitives as well.

\subsection{Motivations for Generalised Quantum Unforgeability}\label{sec:def-motivation}
The first motivation for a new definition lies within the intuitive meaning of unforgeability definition in classical cryptography and its difference within the quantum world. Existential unforgeability is a security notion that formally describes conditions for a function to be \emph{unpredictable} against an adversary who gets access to some query information of that function. To capture this unpredictability at the highest level, an adversary should not be able to produce the output of the function even for a message of his choice. Although to avoid trivial attacks, this message should be ``new", or not equal to any of the queries in the learning phase. This condition can easily be checked by string equality. On the other hand, when translating to the quantum world and giving the adversary quantum access to the oracle, the ``new challenge" can no longer be intuitively defined as before, since the learning phase queries belong to the Hilbert space that can include any desired superposition of classical messages and consequently information from several classical queries. This means an adversary by querying the superposition of all messages can access the output of the function for all of the classical queries in the superposition. Nevertheless, this information needs to be extracted from the quantum state by procedures that are probabilistic in their quantum nature such as measurements. Also, a measurement in the computational basis leads to the collapse of the state into one of the basis states. Hence due to the nature of the measurement and the no-cloning theorem, no more than one classical output can be extracted from such queries by measurements~\cite{holevo1973bounds}. As mentioned in the preliminaries, the first intuitive quantum definition of unforgeability given by \bz\ aims to eliminate trivial attacks by counting the adversary's queries and forcing them to output $q+1$ ``classical" input-output pairs from any $q$ quantum queries. For several reasons, this approach does not properly deal with quantum queries. As also mentioned in~\cite{eurocrypt-2020-30239}, many quantum algorithms need to consume or destroy the quantum states to extract some useful information, such as symmetry in the oracle. As a result, the definition seems to be more restrictive than necessary on a quantum adversary and potentially miss some meaningful attacks. We can also demonstrate this through an example. Assume the adversary issues the following queries to a deterministic oracle and receives the corresponding outputs:

\begin{equation}
    \begin{split}
        & \ket{\phi_1} = \ket{m_1}, \quad \ket{\phi^{out}_1} = \ket{m^{out}_1}\\
        & \ket{\phi_2} = \sigma\ket{m_1} + \gamma \ket{m_2}, \quad \ket{\phi^{out}_2} = \sigma\ket{m^{out}_1} + \gamma \ket{m^{out}_2}
    \end{split}
\end{equation}

where $m_1, m_2$ are bit-strings of length $n$ and $\ket{m^{out}_i}$ are outputs of the oracle's unitary evaluation. And let's assume the oracle is the deterministic quantum oracle corresponding to a MAC algorithm. In that case, $\ket{m^{out}_i} = U_{MAC}\ket{m_i, y} = \ket{m_i, MAC(m_i) \oplus y}$. Now assume that the adversary is trying to forge message $m_2$. One trivial strategy is that the adversary measures the output superposition query and as long as the overlap with $m_2$ is non-negligible, can forge with non-negligible measurement probability. However, as we show in the proof of Theorem~\ref{th:sel-qCM}, there exists more sophisticated quantum attacks for the adversary where they can forge $m_2$, with non-negligible probability even if one excludes the success probability of such trivial attacks. More specifically, for some specific values of $\sigma$ and $\gamma$, there is a quantum emulation attack that produces the output forgery with probability almost 1, and hence will always have a gap with the trivial measurement attacks. We note that such attacks are fairly general and independent of the structure of the underlying deterministic primitive, and hence it is desired for a primitive to resist such attacks. We show later that having schemes that can resists such attacks is possible using randomisation in an effective way, which prevents the adversary to mount emulation or machine learning types of attacks.

Moreover, the \bz\ definition inherently only captures classical challenges and cannot be used for cases where the challenge can be any generic quantum state on an arbitrary basis. Examples of this case are many quantum primitives that use conjugate bases like quantum money~\cite{wiesner1983conjugate,bozzio2018experimental,kumar2019practically} or general input states like quantum PUFs~\cite{arapinis2019quantum}.

In the \bu\ approach, some of the issues of \bz\ have been resolved as this definition does not count the queries and defines the notion of ``new" message in a more natural way using the blind oracle defined in the Preliminary section (Definition~\ref{def:bu}). This definition however, is also only applicable to classical primitives and morally as we will show later, is equivalent to the case where forgeries have no overlap with the adversary's subspace. This definition leads to interesting results, although we believe that some of the attacks we will present, cannot be captured by \bu\ either. 


Following the literature on quantum information, we capture this difference of queries and challenges by a \emph{distance measure} between the respective quantum states. This allows working with natural properties of quantum states irrespective of any assumption on the primitive that generates their output, as well as smoothly capturing all the possible levels of unforgeability as far as the adversary's capabilities go, and hence closing the existing gap. Moreover, having a definition of unforgeability based on quantum distance measures such as fidelity and trace distance allows us to use the quantum information toolkit more easily and intuitively in proofs. Finally, we believe our general unforgeability provides a quantum counterpart for all the different levels of classical unforgeability presented in Table~\ref{table:cunf}. This will also allow us to show which levels of unforgeability and under what assumptions can be achieved in the quantum world.

\subsection{Framework and Formal definitions}\label{sec:game}
Let $\F = (\ES, \E, \V)$ be a classical or quantum primitive with $\ES$, $\E$, and $\V$ being the setup, evaluation, and verification algorithms respectively. Here we focus on classical primitives, and the generalisation can be found in Section~\ref{sec:qGUnf-quantum}. We specify unforgeability as a game between a challenger $\C$ (that models the honest parties) and an adversary $\A$ (that captures the corrupted parties). The adversary's goal is to \emph{closely approximate} the output of the evaluation algorithm $\E$ on a \emph{new challenge} such that it passes the verification with high probability. As we work in the quantum regime, where the adversary has quantum oracle access to the primitive, we adopt the technique of quantum oracles defined in~\cite{boneh2013quantum,boneh2011random} for formalizing quantum query-response interaction between the adversary and the challenger.


The security game considered here consists of several phases. 
First, $\C$ runs the setup algorithm $\mathcal{S}$ to generate the parameters required throughout the game, and instantiates the evaluation oracle $\eO$, the verification oracle $\vO$, and the message space $\M$. The learning phase defines the threat model (we only consider chosen message attacks here). The challenge phase determines the security notion captured by the game. The formal specification of our quantum games is presented in Figure~\ref{fig:game}. But let us first go informally over each phase of the game.


\paragraph{\bf Setup.} In the setup phase, $\C$ generates the parameters required in subsequent phases by running the setup algorithm of the primitive $\F$ on input $\lambda$ (the security parameter), and the oracles are being instantiated accordingly. 

\paragraph{\bf Learning phase.}
In the learning phase, the adversary interacts with the evaluation oracle. Here we only focus on chosen-message attack (cma) security, yet the game can be easily generalised to weaker models such as random-message queries. $\A$ requires the oracle evaluation on any input state $\rho^{in}_i$. The oracle evaluations are handled by $\C$ who issues the requests on $\rho^{in}_i$ to $\eO$ and forwards to $\A$ the respectively received outputs $\rho^{out}_i$, where $i = \{1,\dots, q=poly(\lambda)\}$. We also note that $\A$ can have an internal register $\sigma$ and we allow for creating entanglement between $\A$'s register and output queries. Specifically for classical primitives, each $\rho^{in}_i = \ket{\phi^{in}_i}\bra{\phi^{in}_i}$ where $\ket{\phi^{in}_i} = \sum_{m_i,y_i}\ket{m_i, y_i}$ is usually a pure state with $m_i$ being the message and $y_i$ the ancillary system. If the queries are being generated by $\A$, in most cases it can be assumed that they have the classical information underlying them, while output queries need to be considered as unknown quantum states to the adversary.

\paragraph{\bf Challenge phase.} In this phase, the challenge that the adversary has to respond to, is chosen in three different ways, each corresponding to a specific level of unforgeability. Similar to classical notions of unforgeability, the strongest notion is \emph{existential unforgeability} denoted by $\qEx$ in the game, and whereby the adversary picks the message for which it will produce a forgery. On the other hand, in \emph{selective unforgeability}, denoted  $\qSel$, the adversary picks the challenge but needs to commit to it before interacting with the oracle. Hence in Figure~\ref{fig:game} the selective challenge phase happens before the learning phase. A further way of weakening the unforgeability notion is when the challenge message is chosen by the challenger $\C$ uniformly at random from the set of all the messages.
In any case a classical message $m \in \M$ is selected (for classical primitives) where $\M$ is the set of classical messages.

 
We impose different conditions on the challenge phases which will be formalized later in the guess phase. These conditions prevent the adversary from mounting trivial attacks.

\paragraph{\bf Guess phase.} In this phase, the adversary submits their forgery $t$ for challenge $m$. They win the game if the output pair $(m,t)$ passes the verification algorithm with high probability. In addition, for $\qSel$, the message $m$ should be the same as the message submitted in the challenge phase. Here the condition in the challenge phase that we have mentioned is formally checked. The quantum challenge phase needs to be carefully specified to avoid capturing trivial attacks such as sending one of the previously learnt states as the challenge of the adversary. As a result, we have introduced the notation $m \notmu \rho^{in}$ denoting  $\mu$-distinguishability from all the input learning phase states. When $m$ is a classical bit-string the same condition should hold for the quantum encoding of $m$ into a computational basis i.e. $\ket{m}$ (or $\ket{m, 0}$). Note that the case $\mu=1$ implies the challenge quantum state has no overlap with any of the quantum states queried in the learning phase. 

We emphasize that we do not specify how the challenger could check whether the adversary meets the condition or not. Implementing this check is not crucial for our security analysis, where we only need to be able to characterise the instances that might present a security violation. The key point to note is that this can effectively be checked given a run against a given adversary. Indeed, then $\rho^{in}_i$ and $\rho^{out}_i$ can be characterised allowing proofs of security and exhibition of attacks.


Regarding the verification oracle, for classical primitives the forgery pair $(m,t)$ is classical and the verification oracle $\vO_f$ runs the classical verification algorithm $\V = \mathtt{Ver}(k,m,t,r)$. Here $r$ is the randomness if the primitive is randomised.


\begin{boxfig}{Formal definition of the quantum games $\GCM{\F}{q, c,\mu}$($\lambda,\A$) where $\lambda$ is the security parameter, $q$ the number of queries issued to the evaluation oracle in the learning phase, $\mu$ the overlap allowed between the challenge and previously queries messages, and $c$ the level of unforgeability.}{fig:game}\underline{\emph{The game $\GCM{\F}{q, c,\mu}$($\lambda,\A$)}\footnote{$c\in\{\qEx, \qSel, \qUni\};\ 0 < \mu\le 1$.}}
\medskip\\
{\bf Setup phase:}
\begin{itemize}
\item $\texttt{param} \leftarrow \ES(\lambda)$
\item The oracles $\eO$ and $\vO$ and the message space $\M$ are instantiated given $\texttt{param}$.
\end{itemize}
{\bf Selective challenge phase:}
\begin{itemize}
\item if $c=\qSel$: $\A$ picks $m \in \M$ and sends it to $\C$.
\end{itemize}
{\bf First learning phase:}
\begin{itemize}
    \item $\A$ issues queries $\rho^{in}_1, \dots, \rho^{in}_q$ (where $q=poly(\lambda)$) to $\C$. To each query $\rho^{in}_i$ the challenger $\C$ queries $\eO$ on $\rho^{in}_i$, and forwards the received respective output $\rho^{out}_i$ to $\A$. The adversary can also have an internal register $\sigma$ which may be entangled with the output queries.
\end{itemize}
{\bf Challenge phase:}
\begin{itemize}
\item if $c=\qEx$: $\A$ picks $m \in \M$ and sends it to $\C$.
\item if $c=\qUni$: $\C$ picks $m \overset{\$}{\leftarrow} \M$ uniformly at random and sends $m$ to $\A$
\end{itemize}
{\bf Second learning phase:}
As the \emph{first learning phase}

{\bf Guess phase:}
\begin{itemize}
        \item if $c=\qEx$ OR $c=\qSel$: continue if $m \notmu \rho^{in}$.\footnote{$\notmu \rho^{in}$ denotes at least $\mu$-distinguishability from all the $\rho^{in}_{i}$. For the classical message $m \in \{0,1\}^n$, the condition should hold for $\ket{m}$, i.e. the quantum encoding of $m$ in computational basis.}
        \item $\A$ generates the forgery $t$, and outputs to $\C$ the pair $(m,t) \leftarrow \A(\{\rho_i^{in},\rho_i^{out}\}^q_{i=1},\sigma)$
        \item $\C$ queries the verification oracle: $b \leftarrow \vO (m,t)$
        \item $\C$ outputs $b$
\end{itemize} 
\end{boxfig}\label{fig:game}

We omit the parameter $q$ when we consider arbitrarily polynomially many queries to the evaluation oracle issued by $\A$. We can now formally define \emph{Existential}, \emph{Selective} and \emph{Universal Unforgeability} of primitives as instances of our game as follows.
\medskip{}
\begin{definition}[\eufm]\label{def:euf-qCMA} A cryptographic primitive $\F$ provides $\mu$-quantum existential unforgeability if the probability of any QPT adversary $\A$ of winning the game $\GCM{\F}{\qEx, \mu}(\lambda, \A)$ is at most negligible in the security parameter,
\begin{equation}
Pr[1\leftarrow \GCM{\F}{\qEx, \mu}(\lambda, \A)] \leq \negl(\lambda).
\end{equation}
\end{definition}

We also define a stronger security notion for existential unforgeability which considers any overlap $\mu$.
\begin{definition}[qGEU]\label{def:euf-without-mu} A cryptographic primitive $\F$ provides quantum existential unforgeability if it provides $\mu$-quantum existential unforgeability for all non-negligible $\mu$.
\end{definition}

\begin{definition}[\sufm]\label{def:sel-qCMA} A cryptographic primitive $\F$ provides $\mu$-quantum selective unforgeability if for any $q$ the advantage of any QPT adversary $\A$ of winning the game $\GCM{\F}{q, \qSel, \mu}(\lambda, \A)$ over $P_{ov}(q)$ is at most negligible in the security parameter,
\begin{equation}
Pr[1\leftarrow \GCM{\F}{q, \qSel, \mu}(\lambda, \A)] \leq P_{ov}(q, \mu) + \negl(\lambda).
\end{equation}
We call $P_{ov}(q, \mu)$ the ``overlap probability" describing the probability for trivial attacks via the overlap allowed by the parameter $\mu$.\footnote{Note that by definition $\A$ can always achieve $P_{ov}(q, \mu)$, hence $\A$'s winning probability is always lower-bounded by this value.}
\end{definition}

The need for allowing an adversary to win with probability $P_{ov}(q, \mu)$ is similar to the classical definitions where the adversary is required to boost the success probability from some trivial value such as random guess. Here, by allowing the adversary to create an overlap between the learning phase space and challenge, some unavoidable attacks exist which are independent of the actual primitive at hand, and as such needs to be extracted to characterise the gap between trivial and effective adversaries and hence precisely define a proper distance-based definition.

\begin{definition}[$P_{ov}$ for classical primitives]\label{def:pov-standard-orc}
For all $q$ for all $\mu$
For a classical primitive where the evaluation oracle is a standard oracle $\eO_f$, for any overlap $\mu$ the overlap probability for $q$-query games is equal to $P_{ov}(q, \mu) = 1 - \mu^q$.
\end{definition}

A similar notion for quantum primitive is defined in Section~\ref{sec:qGUnf-quantum}. When selective unforgeability holds for any overlap $\mu$ we say that the primitive is quantum selective unforgeable.

\begin{definition}[qGSU]\label{def:suf-without-mu} A cryptographic primitive $\F$ provides quantum selective unforgeability if it provides $\mu$-quantum selective unforgeability for all non-negligible $\mu$.
\end{definition}

\begin{definition}[\uuf]\label{def:uni-qCMA}
 A cryptographic primitive $\F$ is quantum universally unforgeable if the probability of any QPT adversary $\A$ of winning the game $\GCM{\F}{\qUni}(\lambda, \A)$ is negligible in the security parameter $\lambda$,
\begin{equation}
Pr[1\leftarrow \GCM{\F}{\qUni}(\lambda, \A)] \leq \negl(\lambda).
\end{equation}
\end{definition}

Note that the $\mu$-distinguishability condition is not necessary for Universal Unforgeability, as the challenge is chosen by the challenger, independently of the adversary's queries and the probability is taken over all the choices of the challenge state hence it is no longer meaningful to count for possible overlaps as trivial attacks.

\subsection{Hierarchy and Relationship to other definitions}\label{sec:other-defs}
To demonstrate the generality of our framework and the full context that our results will apply to, we investigate how our definitions formally relate to the previously proposed ones. In particular, we show that \euf\ is equivalent to BU, and hence implies the \bz\ definition (we draw the latter from~\cite{alagic2018unforgeable}). We further formally establish the hierarchy between the different notions of Generalised Unforgeability. In Figure~\ref{fig:hierarchy}, we map out the results presented in this section.

\begin{figure}[h!]
   \centering
     \includegraphics[width=1
     \textwidth]{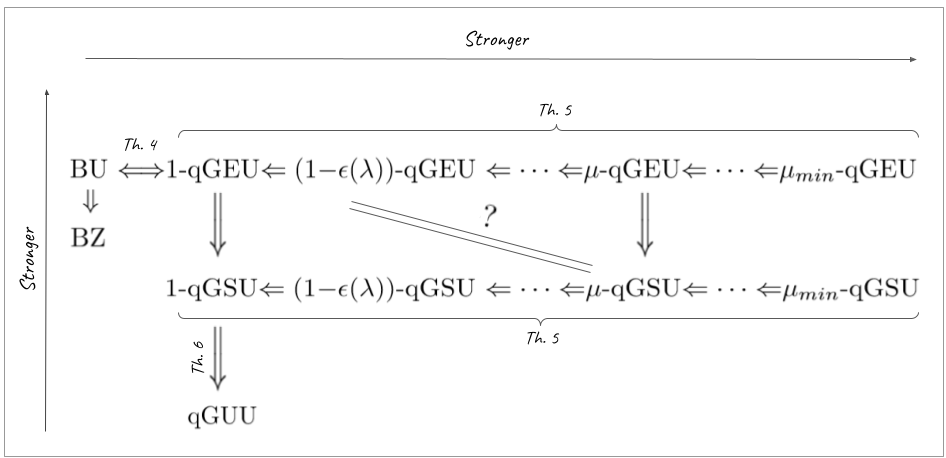}
     \caption{Relationship between different definitions of \emph{Generalised Quantum Unforgeability}, \bu\ and \bz. From down to up and left to right the definitions become stronger. $\epsilon(\lambda)$ is a negligible function in the security parameter and $\mu_{min} = \nonnegl(\lambda)$ is the smallest valid degree for $\mu$. It is unknown whether \sufm\ with smaller $\mu$, implies \eufm\ with bigger $\mu$. }\label{fig:hierarchy}
 \end{figure}

\begin{theorem}\label{th:1guf-bu}
\euf\ is equivalent to BU.
\end{theorem}
\begin{proof}
We show that \euf\ implies BU and vice versa. First, we show that if a scheme is not BU unforgeable against a QPT adversary then it is not \euf\ unforgeable either. Let $\A$ be a QPT adversary who forges a scheme $\F = (\ES, \E, \V)$ with message set $\M = \{0,1\}^n$ in the BU definition. Following the formal definition of BU provided in Definition~\ref{def:bu}, $\A$ selects an $\epsilon$ for which the blinded region $\Be$ is created by selecting each $m \in \M$ at random with an $\epsilon$-related probability. Then there exists a non-empty set $\Be$ for which $\A$ interacts with the blinded oracle associated with it and outputs a pair $(m^*,t^*)$ where $t^* = f(m^*)$ (where $f$ is the classical function of the evaluation $\E$, for instance a \textit{MAC}(.)) such that $\V = Ver_k(m^*,t^*) = acc$, and also the $m^* \in \Be$ with non-negligible probability in $\lambda = poly(n)$. By the definition of the blinding oracle, $\A$ receives a $\ket{\perp}$ for any of the computational basis that are in the blinded region. As a result, we can write $\A$'s input and output queries as follows:
\begin{equation*}
\begin{split}
    & \ket{\phi_i} = \sum_{m_i \not\in \Be} \alpha_i \ket{m_i, y_i} + \sum_{\overline{m}_j \in \Be} \beta_j \ket{\overline{m}_j, y_j} \\
    & \ket{\phi^{out}_i} = \sum_{m_i \not\in \Be} \alpha_i \ket{m_i, y_i\oplus f(m_i)} + \sum_{\overline{m}_j \in \Be} \beta_j \ket{\overline{m}_j, y_j \oplus \perp}
\end{split}
\end{equation*}
Now assuming the quantum encoding of the challenge $m^* \in \Be$ to be $\ket{m^*, 0}$ and the tag/output to be $\ket{m^*, t^*} = \ket{m^*, f(m^*)}$, we can see that $\mbraket{m^*, t^*}{\phi^{out}_i} = 0$ since $m^*$ will have no overlap with the first part of the superposition, and also to the second part due to the blinding. Now, we show that there exists a unitary non-blinding oracle that generates equivalent queries for this scenario. Let $\Ue$ be the unitary evaluation oracle such that $\ket{m^*, t^*} = \Ue\ket{m^*, 0}$, and similarly for all the queries. Due to the unitarity, we have that $\mbraket{m^*, t^*}{\phi^{out}_i} = \bra{m^*, 0}\Ue^{\dagger}\Ue\ket{\phi_i} = 0$. Thus there will also exist an adversary $\A'$ with equivalent queries except that the target forgery will be always orthogonal to the selected challenge. Hence for this adversary, the condition of \euf\ is satisfied. Then by calling $\A$, the adversary $\A'$ can generate an output state $\ket{m^*, t^*}$ that passes the test algorithm with also non-negligible probability. Hence we have shown that \euf\ implies \bu. 


To prove the other way of implication we need to show whenever there is an attack on \euf, then there will also be an attack on \bu\ definition and hence the scheme is also \bu\ insecure. This time we consider $\A$ to be a QPT adversary who wins \euf\ by selecting a challenge state $\ket{m^*, y}$ where the $m*$ is the classical challenge and $y$ is the ancillary register, and querying a set of states $\{\ket{\phi_i}\}^{q}_{i=1}$ s.t. $\forall \ket{\phi_i}: \mbraket{m^*}{\phi_i} = 0$ and $q=poly(n)$. Then by definition, $\A$ can output a $\ket{m^*, t^*} = \Ue\ket{m^*, y}$ that passes the test algorithm with non-negligible probability. Now an adversary $\A'$ calls $\A$ to win the \bu\ with non-negligible probability.

At this stage we recall the Theorem~\ref{th:bu} and we show that an $\A'$ satisfies the conditions of this theorem. Let us write the learning phase queries in the computational basis as follows:
\begin{equation}
    \ket{\phi^{out}_i} = \sum^{d}_{j=1} \alpha_{i,j} \ket{b_j}
\end{equation}
where $\{\ket{b_j}\}^{d}_{j=1}$ is the set of computational bases spanning the effective learning phase subspace. Now we create a non-empty set $R$ by selecting each $x \in \M$ as follows
\begin{equation}
    R = \{x \in \M : \ket{x} \neq \ket{b_j}_X \}
\end{equation}
Where $\ket{b_j}_X$ denotes the input register of the full basis. 
Note that $R$ will always be non-empty as the basis set will only cover a polynomial-size subspace of the whole Hilbert space of messages. Moreover, since $\A'$ includes $\A$ and $m^*$ has no overlap with any of the input queries, it will also have no overlap with the input register of the output queries. As a result, $R$ has at least one element. Hence the set of all input elements that have non-zero overlap with the queries and the elements included in $R$ have no intersection. This shows that $supp(\A)\cap R = \emptyset$ if the support is defined for the oracle $\Ora_f$ for a fixed randomly picked classical function $f$ (or key $k$) during the game. Thus we also have $supp(\A')\cap R = \emptyset$ and $m^* \in R$. Nevertheless, in~\cite{alagic2018unforgeable} has been mentioned that the support is taken to be the union of the support of all the queries over the choice of the function. In this case we can also redefine our set, and the queries of $\A'$ such that it satisfies the condition of the theorem respectively. We take the set $R'$ to only include one element which is the forgery message $m^*$. As in the \euf\ the function (or the key for the keyed functions) is selected at random in the setup phase, the success probability of $\A$ is inherently taken over the choice of the function. Then $\A'$ queries all the queries of $\A$ for any randomly selected $f$ during the experiment. For any other functions, excludes any queries for which the support will include $m^*$. Now we can see that $\A'$ can output a valid pair $(m^*, t^*)$ by measuring $\ket{m^*, t^*}$ in the computational basis with probability 1 while $supp(\A')\cap R' = \emptyset$ and $m^* \in R'$. Hence $\A'$ breaks the \bu\ unforgeability and we have shown that \bu\ implies \euf. This mutual implication shows that these definitions are equivalent and the proof is complete.\qed
\end{proof}

From the above theorem and the equivalence of \bu\ and \bz\ against classical adversaries we derive the following corollary.
\begin{corollary}
$\euf \equiv \bu \equiv \bz$ against classical adversaries.
\end{corollary}

Next, we establish the relation between different instances of our game-based definition. First, we emphasise that as expected for both existential and selective unforgeability, the definitions become stronger when decreasing the $\mu$ parameter from 1 and hence \eufm\ implies \euf. 

\begin{theorem}\label{th:mu-smaller-stronger}
If $\mu_1 \leq \mu_2$ then $\mu_1$-qGEU ($\mu_1$-qGSU) implies $\mu_2$-qGEU ($\mu_2$-qGSU)
\end{theorem}
\begin{proof}
The proof is straightforward for qGEU. Let $\A$ win against $\mu_2$-qGEU, Let $\A'$ be the adversary who wants to attack $\mu_1$-qGEU. $\A'$ queries the same learning phase queries as $\A$ and then calls $\A$. Since $\mu_1 \leq \mu_2$ any two states that are $\mu_2$-distinguishable are also $\mu_1$-distinguishable, then the challenge of $\A$ will necessarily satisfy the condition for $\mu_1$-qGEU. Then $\A'$ can also win the game with non-negligible probability. For $\mu$-qGSU the distinguishability argument is similar, although there is also the $P_{ov}$ probability that is function of $\mu$. Thus we need to show the following:
\begin{equation*}
    Pr[1\leftarrow \GCM{\F}{\qSel, \mu_2}(\lambda, \A)] - P_{ov}(\mu_2) \geq  Pr[1\leftarrow \GCM{\F}{\qSel, \mu_1}(\lambda, \A)] - P_{ov}(\mu_1)
\end{equation*}
Which is also equivalent to showing the following statement:
\begin{equation*}
    Pr[1\leftarrow \GCM{\F}{\qSel, \mu_2}(\lambda, \A)] - Pr[1\leftarrow \GCM{\F}{\qSel, \mu_1}(\lambda, \A)] \geq P_{ov}(\mu_2) - P_{ov}(\mu_1)
\end{equation*}
The LHS of the inequality is always positive due to the above distinguishability argument, and the $P_{ov}$ is always a non-increasing function of $\mu$ for both types of primitives. Take the $P_{ov}$ for the classical primitives for instance, which is equal to $1-\mu^q$. Therefore, the RHS of the inequality will be equal to $\mu_1^q - \mu_2^q$ which is always a non-positive value as $\mu_1 \leq \mu_2$. Then the above inequality holds and the theorem has been proved. \qed 
\end{proof}

Furthermore, it is easy to observe that for any given $\mu$, \eufm\ implies \sufm. This is due to the fact that if the adversary wins the game by committing to their favourite message before the learning phase, they will necessarily win when picking the message after the learning phase.

Universal unforgeability is also intuitively weaker than existential unforgeability similarly to their classical counterpart. This holds, despite the winning condition for these two instances being very different. In universal unforgeability, the adversary wins only if they win the game on average over all the different randomly picked messages. Since in our case, we are only interested in QPT adversaries, and as the universal definition is not parameterised by $\mu$, it is not obvious that $\uuf$ is weaker than \sufm. In the following theorem, we formally establish the implication. We prove the theorem for \suf\ which in turn implies \sufm\ for any $\mu$.

\begin{theorem}\label{th:suf-uuf}
\sufm\ implies \uuf.
\end{theorem}
\begin{proof}[sketch]
The full proof can be found in Supplementary Materials~\ref{app:proof-suf-uuf}. Here we present the key ideas of the proof. We show if there exists an adversary $\A$ that wins the \uuf\ game then \euf\ (\suf) also breaks and the implication to \eufm\ (\sufm) is straightforward. First, we show that the distinguishability condition for $\mu=1$ can be satisfied. Thus we write the winning probability of $\A$ as the combination of probabilities of winning with respect to the selected message being orthogonal to the learning phase or not:
\begin{equation}
\begin{split}
    \underset{x \in \M}{Pr}[1\leftarrow \A(x)] & = \underset{x \in \M'}{Pr}[1\leftarrow \A(x)]Pr[x \in \M'] + \underset{x \not\in \M'}{Pr}[1\leftarrow \A(x)]Pr[x \not\in \M'] \\
    & = \nonnegl(\lambda)
\end{split}
\end{equation}
where $\M'$ is the set of all the challenges with no overlap with the learning-phase states. By calculating this probability we show that $\underset{x \in \M'}{Pr}[1\leftarrow \A(x)]$ is also non-negligible. In the second part of the proof we show that as long as the previous average probability holds, we can always construct an efficient adversary $\A'$ that uses $\A$ to win the selective unforgeability game. We prove this by partitioning the space of $\M'$ into equal polynomial-size subspaces and show that if the average probability over $\M'$ is non-negligible, then $\A'$ can always win the \euf\ game by randomly picking one of the subsets to pick the message from, as there will exist at least one message that allows $\A$ to win the game with non-negligible probability. As a result, $\A'$ wins the game with non-negligible probability. \qed
\end{proof}
\section{Possibility and Impossibility results}
\subsection{Generalised Existentially Unforgeable Schemes}
\label{sec:ex-unf}
In this section, we turn our attention to \euf. First, we show a general and intuitive, yet important no-go result for \eufm\ that is, no classical primitive (deterministic nor randomized) can satisfy this level of unforgeability for any $\mu \neq 1$. This result states that \euf, which is equivalent to \bu\ as shown in the previous section, is the strongest notion of existential unforgeability that any classical primitive can achieve.


\begin{theorem}[No classical primitive $\F$ is \eufm~ secure]\label{th:ex-qCM} For any classical primitive $\F$ and for any $\mu$ such that $\mu \leq 1 - \frac{1}{2^n}$, there exists a QPT adversary $\A$ such that
\begin{equation}
    Pr[1\leftarrow \GCM{\F}{\qEx, \mu}(\lambda, \A)] = \nonnegl(\lambda).
\end{equation}
\end{theorem}
\begin{proof}
There exists a simple superposition attack that breaks \eufm. Let $\A$ issue only one query which is the uniform superposition of all the inputs, which leads to an output of the form $\frac{1}{\sqrt{2^n}}\sum_{m} \ket{r}_{\Ora}\ket{m, f(m;r)}$. Then by measuring the first part of the register in the computational basis, the state will collapse to one of the basis and the adversary is able to produce a valid message-tag pair for a classical message with a negligible overlap with the learning phase. Hence $\A$ can always win the game for any any $\mu \leq 1 - \frac{1}{2^n}$. \qed
\end{proof}


Nevertheless, it is still possible to have schemes that are \euf\ secure through the following positive result:
\begin{theorem}\label{th:qprf-1euf} $\qprf$s are \euf\ (\suf) unforgeable.
\end{theorem}
\begin{proof}
This is a straightforward result via equivalence of \euf\ to \bu ~and Corollary 4 in~\cite{eurocrypt-2020-30239}, where it is shown that $\qprf$s are $\bu$ secure.\qed
\end{proof}

\subsection{Generalised Selectively Unforgeable Schemes}
\label{sec:sel-unf}
In this section, we establish results for \sufm\ which restricts the adversary in two ways. First, by requiring the adversary to commit to the challenge before the learning phase, we prevent the adversary to pick any post-measurement state as their forgery challenge. Second, by subtracting the probability of any potential trivial attack, especially for classical primitives, from the winning probability of the game, we make the probability bounds tighter for the adversary. We show that defining unforgeability in this way leads to non-trivial results and establish a separation between randomised and non-randomised constructions. 

\subsubsection{Non-randomised Schemes}
We show a general impossibility result using the \emph{quantum emulation attack} introduced in~\cite{arapinis2019quantum}. Here we only show this no-go result for classical non-randomised primitives to avoid repetitions, but the same result holds for quantum construction too. 

\begin{theorem}[No classical (or quantum) non-randomised primitive $\F$ is \sufm\ secure]\label{th:sel-qCM} For any classical/quantum primitive $\F$ and for any $\mu$, in the range $\frac{1}{4} + \nonnegl(\lambda)\leq \mu \leq 1-\nonnegl(\lambda)$, there exists an effective QPT adversary $\A$ such that
\begin{equation}
Pr[1\leftarrow \GCM{\F}{q(\lambda), \qSel, \mu}(\lambda, \A)] - P_{ov}(q(\lambda), \mu) = \nonnegl(\lambda).
\end{equation}
\end{theorem}

\begin{proof}[sketch]
We show the proof for classical primitives but the same attack and results also holds for quantum primitives. We show that there exists a QPT adversary $\A$ who can win the game with non-negligible probability for any $\mu$ except when it is negligibly close to 0 or 1. A more detailed version of the proof is given in the Supplementary Materials~\ref{app:proof-selective-unf-nogo}. The attack we present is an emulation attack based on the universal quantum emulator~\cite{marvian2016universal}. First $\A$ picks any two messages $m, m' \in \M$ and sets $m$ as the challenge. Then $\A$ queries the states $\ket{\phi_1} = \ket{m', 0}$ and $\ket{\phi_r} = \sqrt{1-\gamma^2}\ket{m', 0} + \gamma \ket{m, 0}$ from $\eO_f$,
where $\gamma$ is a real value such that $0 \leq \gamma \leq \sqrt{1-\mu}$ and such that the distinguishability condition of the \sufm\ game is satisfied. After the learning phase, $\A$'s output state is $\sigma_{out} = \ket{\phi^{out}_1}\otimes\ket{\phi^{out}_r}$ where $\ket{\phi^{out}_1} = \Ue \ket{\phi_1}$ and $\ket{\phi^{out}_r} = \Ue \ket{\phi_r}$. Followed by the fidelity analysis of the attack algorithm given in Supplementary Materials~\ref{app:proof-selective-unf-nogo}, we show that the success probability of $\A$ in producing the output of $m$ i.e. $f(m)$ is $Pr[1\leftarrow \GCM{\F}{2, \qSel, \mu}(\lambda, \A)] = \gamma^2(1 + 4(1-\gamma^2)^2)$. We also note that for $\gamma = \frac{1}{\sqrt{2}}$, there exists an emulation with success probability 1. 
Also, we let $\A$ to set $\gamma$ to the maximum value allowed by the overlap condition i.e. $\gamma = \gamma_{max} = \sqrt{1-\mu}$. Finally, we need to subtract the $P_{ov}$ from this probability for the adversary to be effective. For this attack we have $q=2$ and the $P_{ov}(2, \mu) = 1 - \mu^2$ according to Definition~\ref{def:pov-standard-orc}. Thus we have
\begin{equation}
    Pr[1\leftarrow \GCM{\F}{q, \qSel, \mu}(\lambda, \A)] - P_{ov}(2, \mu) =  \mu(1-\mu)(4\mu - 1) = \nonnegl(\lambda)
\end{equation}
which concludes the proof. \qed
\end{proof}

Despite the above no-go result, \qprf s still provide \suf\ security, as mentioned in Theorem~\ref{th:qprf-1euf}. However, the above theorem shows a fundamental vulnerability of any non-randomised classical primitive against forgeries, since the only way to ensure the security of primitives against such effective attacks is to guarantee that the adversary's forgery message is orthogonal to their learning subspace by relying on the device implementation which is in contradiction with the whole motivation of obtaining security against more powerful quantum adversaries, to begin with. More precisely, our Theorem~\ref{th:sel-qCM} shows that non-randomised MAC schemes such as HMAC and NMAC do not satisfy existential nor selective unforgeability except for $\mu = 1$ and hence are always vulnerable against more powerful quantum adversaries implementing superposition attacks. At this point, we go back to the same example that we have presented in Section~\ref{sec:def-motivation}, which illustrates more clearly why the current definition and the quantum emulation class of attacks shows a forgery that clearly needs to be prevented. We present a slightly different attack to the one exhibited in the proof of Theorem~\ref{th:sel-qCM} but that makes even more obvious the need for our generalised definition.

\begin{example}\label{example:qea-three-states}
Let $\A$'s state after the learning phase be $\sigma_{in} = \ket{\phi^{in}_1}\otimes\ket{\phi^{in}_r}^{\otimes 2}$ and $\sigma_{out} = \ket{\phi^{out}_1}\otimes\ket{\phi^{out}_r}^{\otimes 2}$ where the query states have been chosen as follows:
\begin{equation}
    \ket{\phi_1} = \ket{m_1, 0} \quad \ket{\phi_r} = \delta \ket{m_1, 0} + \gamma \ket{m_2, 0} + \gamma \ket{m_3, 0}
\end{equation}
Where due to normalisation $|\delta|^2 + 2|\gamma|^2 = 1$, although we pick the $\delta = \sqrt{1-2\gamma^2}$ and $\gamma$ to be real values for simplicity, thus $\gamma^2 \leq \frac{1}{2}$. Also note that $\A$ has two identical copies of $\ket{\phi^{out}_r}$. The attack consists of running two separate emulations for $\ket{m_2, 0}$ and $\ket{m_3, 0}$.

Let $\ket{\phi_r}$ be the reference state for the emulation, and the target state to be $\ket{\psi} = \ket{m_2, 0}$ or $\ket{\psi} = \ket{m_3, 0}$. Note that as $\ket{\phi_1} = \ket{m_1, 0}$ is orthogonal to both states and the reference state is symmetric with respect to them, the emulation's fidelity will be the same for both these states. Relying on Theorem~\ref{th:qe-fins}, the output state of the QE algorithm with only one block will be:
\begin{equation}
\begin{split}
     \ket{\chi_f} =& \mbraket{\phi_r}{\psi} \ket{\phi_r}\ket{0} + \ket{\psi}\ket{1} - \mbraket{\phi_r}{\psi}\ket{\phi_r}\ket{1}-2\mbraket{\phi_1}{\psi}\ket{\phi_1}\ket{1} \\ & +2\mbraket{\phi_r}{\psi}\mbraket{\phi_r}{\phi_1}\ket{\phi_1}\ket{1}.
\end{split}
\end{equation}
Note that $|\mbraket{\phi_1}{\psi}| = 0$ and $|\mbraket{\psi}{\phi_r}|^2 = \gamma^2$ and $|\mbraket{\phi_1}{\phi_r}|^2 = 1 - 2\gamma^2$. Then according to Theorem~\ref{th:qe-fidel}, the fidelity of the emulation for both states is:
\begin{equation}
    F(\ket{\omega}\bra{\omega}, \Ue \ket{\psi}\bra{\psi} \Ue^{\dagger}) \geq \gamma^2(1+4(1-2\gamma^2)^2)
\end{equation}

Now we need to compare this probability with the $P_{ov}$ probability which is $P_{ov}(3, \mu) = 1-\mu^3$ since the size of the learning phase includes 3 queries. We write the effective success probability of the adversary as: 
\begin{equation}\label{eq:prob-win-selective-two-messages}
\begin{split}
    Pr_{forge}[\A(m_2)] & = Pr_{forge}[\A(m_3)] = Pr[1\leftarrow \GCM{\F}{3, \qSel, \mu}(\lambda, \A)] - P_{ov}(3, \mu) \\
    & = \gamma^4(1+4(1-2\gamma^2)^2)^2 - (1-\mu^3)
\end{split}
\end{equation}



Finally, we need to do a functional analysis of the above probability to see in which cases it becomes non-negligible. First, we note that the success probability of the emulation attack is not greater than the trivial success probability for all the values of $\mu$ which shows that if we allow for too much overlap, the trivial attack already has a very high probability which is higher than the emulation's fidelity in this case. Next, since the highest allowed overlap is achieved when $1 - \mu = \gamma^2$, we substitute the variable $\mu$ with $1 - \gamma^2$ to find the degrees of $\mu$ for which an effective adversary exists. Hence we rewrite the winning probability of the equation~\ref{eq:prob-win-selective-two-messages} as follows:
\begin{equation}
    Pr_{forge}[\A(m_2 \vee m_3)] = \gamma^2(1+4(1-2\gamma^2)^2) - (1 - (1-\gamma^2)^3)) = \gamma^2(2 - 5 \gamma^2 + 3 \gamma^4)
\end{equation}
Noting that the valid range for $\gamma$ is $0\leq \gamma \leq \frac{\sqrt{2}}{2}$, we plot the above function as it is shown in Figure~\ref{fig:example-prob-plot} and we can see that there is exist a valid range for $\mu$ such that the above forgery attack happens with non-negligible probability.

But more importantly, now having access to two copies of the reference state, the adversary can actually run the emulation attack twice, and produce the outputs of both $m_2$ and $m_3$ at the same time, with non-negligible probability. Thus for these values of $\mu$, we have presented an adversary who can produce effective forgery for three classical messages $m_1$, $m_2$ and $m_3$ (Note that the first learning phase query is $\ket{m_1, 0}$ which is basically a classical query and as a result, $\A$ will always have the output for $m_1$) from a classical query, and two copies of the same quantum state which shows an intuitive forgery, especially that the presented attack is independent of the size of the messages and the dimensionality of the Hilbert space of the oracle. This sort of attacks cannot be captured in the definitions of unforgeability that count the queries, such as \bz. Nevertheless, our approach in defining the notion of unforgeability is capable of showing such vulnerabilities against strong quantum adversaries.
\begin{SCfigure}
     \includegraphics[width=0.50\textwidth]{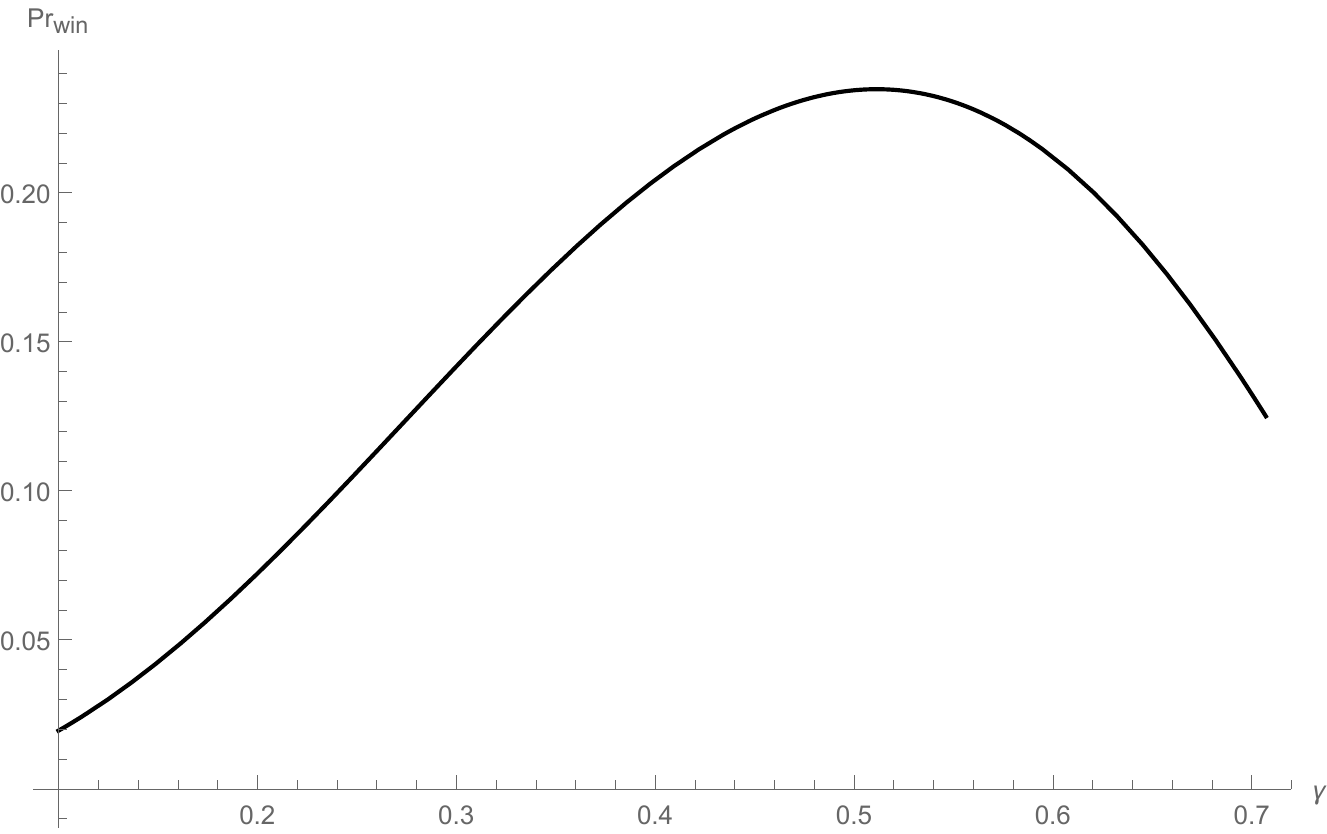}
     \caption{\small The winning probability of $\A$ to forge classical messages $\{m_2,m_3\}$ with the emulation attack. $\gamma$ represents the overlap between the learning phase query and the target message.}\label{fig:example-prob-plot}
\end{SCfigure}
\end{example}

\subsubsection{Randomised Schemes:}
In this section, we explore how to defend against general superposition adversaries, \emph{i.e.} that are allowed to exploit overlaps between previously queried messages and the target message. We show that selective unforgeability can be achieved in such a setting, by effective randomization. Concretely, we present a randomized construction for classical primitives that satisfies qGSU (\sufm\ for any $\mu$). The key ingredient that allows this construction to be secure is that the randomization has been used in an effective way such that the adversary is prevented from creating a known subspace for a specific unitary, even though they can query the challenge message in superposition. First, we formalise the desired characteristic for the family of the classical functions used in our construction.

\begin{definition}[Inter-function independent family:]\label{def:computational-function-pairwise} Let $F_k: \K \times \X \rightarrow \Y$ be a keyed family of functions with domain $\X$ and range $\Y$, where $\X=\{0,1\}^n$ and $\Y=\{0,1\}^m$. We say $F_k$ is an \emph{inter-function (pairwise) independent family} if for any efficient PPT adversary $\A$ and any two functions $F(k, .)$ and $F(k', .)$ picked uniformly at random from $F_k$, the probability of $\A$ finding an $x \in \X$ such that $F(k, x) = F(k', x)$, is negligible in the security parameter, i.e. the following condition should hold:
\begin{equation}\label{eq:computational-function-pairwise}
    \underset{k, k' \leftarrow \K}{Pr}[x \leftarrow \A(1^{\lambda}) \wedge F(k, x) = F(k', x)] = \negl(\lambda)
\end{equation}
\end{definition}

Now we show that a \prf\ family satisfies the above condition.
\begin{lemma}\label{lemma:prf-inter-function-ind}
A \prf\ is an inter-function independent family.
\end{lemma}
\begin{proof}
We want to show that any two randomly selected functions from a PRF family, satisfy the required pairwise-independency property of Definition~\ref{def:computational-function-pairwise}. Let $F_k: \K \times \X \rightarrow \Y$ be a PRF family of functions where $|\X| = 2^n$ and $|\Y| = 2^m$. We want to show that there is no efficient adversary that can find an $x$ such that $F(k, x) = F(k', x)$ for any two different, randomly picked keys $k, k'$.
We prove by contradiction. We assume that $F_k$ is a PRF but there exist an efficient adversary $\A$ that can find at least one $x \in \X$ such that for any two randomly picked functions from $F_k$ we have:
\begin{equation}\label{eq:prf-pairwise-false}
    \underset{k, k' \leftarrow \K}{Pr}[x \leftarrow \A(1^{\lambda}) \wedge F(k, x) = F(k', x)] = \nonnegl(\lambda).
\end{equation}

Now we construct a new family of functions from $F_k$ which is a PRF. Let $F'_{k,k'}: \K^2 \times \X \rightarrow \Y$ be constructed as follows:
\begin{equation}
    F'((k,k'), x) = F(k, x) \oplus F(k', x)
\end{equation}
It is a well-known example in the literature that if $F_k$ is a PRF, then $F'_{k,k'}$ is also a PRF.
Now we show that if the equation~(\ref{eq:prf-pairwise-false}) holds, then there also exist an adversary who can distinguish $F'((k,k'), x)$ form truly random function. Let $\A'$ query the same $x'$ that has been found by $\A$. If $\A'$ queries $F'((k,k'), x)$, since $F(k, x') = F(k', x')$ with non-negligible probability, then the queries to $F'((k,k'), x)$ on $x'$ should return $0^m$. On the other hand the queries to the truly random function will return random bit-strings. As a results, $\A'$ can distinguish $F'((k,k'), x)$ from a truly random function which is a contradiction and hence we have proved that \prf\ satisfies the Definition~\ref{def:computational-function-pairwise}. \qed
\end{proof}

We can now give our construction based on \prf s or more generally, based on any family of classical functions satisfying the Definition~\ref{def:computational-function-pairwise}.

\begin{construction}\label{const:classical-random-const-prf}
Let $F: \K \times \X \rightarrow \Y$ be a \prf\ (or any other family satisfying Definition~\ref{def:computational-function-pairwise}). Let $\R = \K = \{0,1\}^l$ be the randomness space. And let $\lambda$ be the security parameter and $l$ be polynomial in $\lambda$. The construction is defined by the following key generation algorithm, keyed evaluation algorithm, and keyed verification algorithm:
\begin{itemize}
    \item{\bf Key generation:} The secret key is picked uniformly at random from $\K$: $k \xleftarrow{\text{\$}} \K$
    \item{\bf Evaluation:} The evaluation under key $k$ on input $m$ picks randomness $r$ and applies $F(k\oplus r, \cdot)$ to $m$. Note that when responding to a quantum query, the same randomness is used for all the states of the superposition:
    \begin{itemize}
            \item On input $m\in \X$:
            \item $r \xleftarrow{\text{\$}} \R$
            \item Return $F(k \oplus r, m)||r$ 
    \end{itemize}
    \item{\bf Verification:} The verification under key $k$ of a pair $(m, (t, r))$, runs the evaluation algorithm on $m$ under $k$ with randomness $r$, and checks equality with $t$.
    \begin{itemize}
        \item On input $(m, (t, r))\in \X \times (\Y \times \R)$:
        \item If $F(k \oplus r, m) = t$ return $\top$, otherwise return $\bot$
    \end{itemize}
\end{itemize}
\end{construction}

Now we show that the construction satisfies \sufm\ security. 

\begin{theorem}\label{th:classical-random-const-prf-secure}
Construction~\ref{const:classical-random-const-prf} is qGSU secure.
\end{theorem}

\begin{proof}
We assume there exists a QPT adversary $\A$ who plays the \sufm\ game where the evaluation is according to Construction~\ref{const:classical-random-const-prf}, and wins with non-negligible probability in the security parameter \emph{i.e.} $\A$ wins the game by producing a valid tag $t^*$ for their selected message $m^*$ and randomness $r^*$ with the following probability:
\begin{equation}\label{eq:adv-prob-wining-suf}
    Pr[1\leftarrow \GCM{\F}{q, \qSel, \mu}(\lambda, \A)] - P_{ov}(q_r, \mu) = \nonnegl(\lambda)
\end{equation}
Where the verification algorithm checks if $F(k \oplus r^*,m^*) = t^*$. We introduce the following games:
\begin{itemize}
    \item \textbf{Game 0.} This game is the \sufm\ for Construction~\ref{const:classical-random-const-prf}, where $F(k \oplus r, .)$ is picked from $F$.
    \item \textbf{Game 1.} This game is similar to Game 1, except that $\A$ needs to produce forgery for a $r^*$ which is one of the previously received random values of $\{r_i\}^q_{i=1}$ in the learning phase.
\end{itemize}

First, it is straightforward that the probability of the adversary in winning \sufm\ in Game 0, is at most negligibly higher than winning Game 1. Since $r_i$ in both cases have been picked independently and uniformly at random and the probability of producing a forgery for a specific function with no query is negligible. Thus Game 0 and Game 1 are indistinguishable.

Now we recall the quantum random oracle for this construction. Let $\reO_c$ be the random oracle for both games:
\begin{equation}
    \reO_c: \sum_{m,y} \alpha_{m,y} \ket{r}_{\Ora}\ket{m,y} \rightarrow \sum_{m,y} \alpha_{m,y}\ket{r}_{\Ora}\ket{m, y \oplus (F(k \oplus r,m) || r)}
\end{equation}

Note that in each query a new function has been picked from $F$, but it is the same for all the messages in the superposition for that query.

Now we use the inter-function (pairwise) independent property of the family $F$.
The construction requires the $F$ to be a \prf\ family which is inter-function independent according to Definition~\ref{def:computational-function-pairwise}, for two randomly selected keys. Now we need to also show that $F(k\oplus r, .)$ is a \prf\ as well, with a key $k$ and any randomly selected randomness $r$, and as a result we can use the inter-function independent property. This is clearly the case as the key $k$ and any randomness $r$ have been picked independently at random and if there exist a non-negligible advantage for the adversary to distinguish a $F(k\oplus r, .)$ from a truly random function for a value of $r$, there also exist an equivalent non-negligible advantage to distinguish a $F(k', .)$ where $k' = k \oplus r$ is a key selected uniformly at random. This is still the case even if the value $r$ becomes public after the experiment. This is in contrast with the assumption that the family is PRF, hence we conclude that $F(k\oplus r, .)$ is a \prf. Now we can rely on the Lemma~\ref{lemma:prf-inter-function-ind} that $F(k\oplus r, .)$ also satisfies the inter-function independent property and the following holds for each of the two functions drawn in any of the two queries:
\begin{equation}\label{eq:inter-func-inside-proof}
    \underset{i,j (i\neq j)}{Pr}[x \leftarrow \A(1^{\lambda}) \wedge F(k \oplus r_i, x) = F(k \oplus r_j, x)] = \negl(\lambda)
\end{equation}

As a result, we show that the adversary can at most span a one-dimensional subspace of each $U_{k\oplus r}$. To show this we will calculate the probability of $\A$ in spanning at least a 2-dimensional common subspace from two different queries. This means that $\A$ needs to find at least two bases mapping to the same 2-dimensional subspace in the output Hilbert space. Moreover, we exclude that part of $\A$'s register that contains the classical value of the randomness in order to only capture the Hilbert space of each $U_{k\oplus r}$. Thus let the input bases be denoted by $\ket{b} = \ket{m, z}$ where $z$ is a subset of $y$ excluding the space for the randomness, for a specific $m$. Let $\ket{e_i} = U_{k\oplus r_i}\ket{b} = \ket{z \oplus F(k \oplus r_i, m)}$ and $\ket{e_j} = U_{k\oplus r_i}\ket{b} = \ket{z \oplus F(k \oplus r_j, m)}$ be the output states from two different queries. For these output bases to have overlap, the two functions $F(k \oplus r_i,.)$ and $F(k \oplus r_j, .)$ need to return the same classical output with high probability. Although from equation~(\ref{eq:inter-func-inside-proof}), we have that the probability of finding such inputs that leads to a common basis is negligible:

\begin{equation}
\begin{split}
    & \underset{i,j (i\neq j)}{Pr}[\{\ket{e_i}, \ket{e_j}\} \leftarrow \A(1^{\lambda}) \wedge \mbraket{e_i}{e_j} \neq 0] \\
    & = \underset{i,j (i\neq j)}{Pr}[\ket{b} \leftarrow \A(1^{\lambda}) \wedge \bra{b}U^{\dagger}_{k\oplus r_i} U_{k\oplus r_j}\ket{b} \neq 0] \\
    & = \underset{i,j (i\neq j)}{Pr}[\ket{b} \leftarrow \A(1^{\lambda}) \wedge \mbraket{z \oplus F(k \oplus r_i, m)}{z \oplus F(k \oplus r_j, m)} \neq 0] \\
    & = \underset{i,j (i\neq j)}{Pr}[m \leftarrow \A(1^{\lambda}) \wedge F(k \oplus r_i, m) = F(k \oplus r_j, m)] = \negl(\lambda)
\end{split}
\end{equation}
This means that finding an even 2-dimensional common subspace between the different unitaries of the set is hard for $\A$. Also since unitaries are distance preserving operators, this property holds for any sets of orthonormal basis, not necessarily the computational basis. Thus by selecting a uniformly random function for each query, we have shown that no more than a one-dimensional subspace can be spanned for each specific unitary. 

Now we calculate the upper-bound of $\A$'s probability from a single query to a fixed unitary $U_{k \oplus r^*}$ which we denote by $U^*$ for simplicity. We recall that this query should be $\mu$-distinguishable with the quantum encoding of $m^*$. Without loss of generality, let us write $\A$'s selected query for $r^*$ as follows:
\begin{equation}
\begin{split}
        & \ket{\phi_{r^*}} = \alpha\ket{m^*, z, 0} + \beta\ket{\Omega}\ket{0},\\
        & \ket{\phi^{out}_{r^*}} = (\alpha\ket{m^*, z \oplus F(k \oplus r^*, m^*)} + \beta U^*\ket{\Omega})\ket{r^*}
\end{split}
\end{equation}
where $\ket{\Omega}$ is a normalised state that includes a superposition of a set of messages $m \neq m^*$ and as a result $\mbraket{m^*,z}{\Omega} = 0$ and $\A$ sets the second part of the register to $0$, such that the output randomness is a separable state and it can be excluded in the rest of the proof. Due to the fact that $U^*$ is unitary, we know that $\bra{m^*, z\oplus F(k \oplus r^*, m^*)}U^*\ket{\Omega} = 0$ and hence the probability of outputting  $F(k \oplus r^*, m^*)||r^*$ from $\ket{\phi^{out}_{r^*}}$ is at most the probability of measuring it in the computation basis which is $|\alpha|^2$. This probability is maximum when $|\alpha| = |\alpha_{max}|$ which is when $\A$ uses the maximum allowed overlap of size $\sqrt{1-\mu}$. Hence we have:
\begin{equation}
    Pr[1\leftarrow \GCM{\F}{q_r,\qSel, \mu}(\lambda, \A)] \leq 1 - \mu
\end{equation}
But on the other hand we have $P_{ov}(1, \mu) = 1-\mu$, which is the lower bound for $P_{ov}(q, \mu)$, and also since there is only one query to each function selected by each $r$, and equation~(\ref{eq:adv-prob-wining-suf}) states that this probability is negligibly higher that $1 - \mu$. Thus we have reached a contradiction that concludes our proof. \qed
\end{proof}

Theorem~\ref{th:classical-random-const-prf-secure} shows that in addition to \prf , \qprf s can also be used in the construction to achieve selective unforgeability. Nevertheless, we have also provided a separate security proof for the \qprf\ family that does not need the Definition~\ref{def:computational-function-pairwise}. This proof can be found in Supplementary Materials~\ref{app:proof-randomised-suf}.

\subsection{Generalised Universally Unforgeable Schemes}
\label{sec:uni-unf}
In this section, we further weaken the notion of unforgeability and provide a generally positive result for universal unforgeability. We recall that here the adversary receives a challenge picked by the challenger uniformly at random from the full message space.

Despite the fact that universal unforgeability is a weaker notion, it is a practical definition and sufficient for many protocols, especially for the case of quantum primitives and quantum protocols~\cite{arapinis2019quantum,doosti2020client}. Here as we mostly focus on classical primitives, we mention this definition for the sake of completeness and we specify with the following theorem that \qprf\ is enough to achieve \uuf, due to the hierarchy of our definition.


\begin{theorem}\label{th:qprf-uuf}
\qprf s are \uuf\ secure.
\end{theorem}
\begin{proof}
This is a direct implication of Theorem~\ref{th:qprf-1euf} where we have proved that \qprf s are \suf\ secure and Theorem~\ref{th:suf-uuf} showing that \suf\ implies \uuf. \qed
\end{proof}


In the Supplementary Materials~\ref{ap:uni-adaptive} we also give a general no-go result for the \uuf\ security of quantum primitive against universal but adaptive attack model, where the adversary can issue learning phase queries after receiving the random challenge that is selected by the challenger. We show that in this case there also exists an interesting entanglement-based attack which leads to breaking \uuf\ against this adversarial model.

\section{Generalization of the definition to quantum primitives}\label{sec:qGUnf-quantum}
In this section, we show how our framework can also capture the natural notion of unforgeability for quantum primitives. We refer to quantum primitives as the primitive where the input message space and the output space are both Hilbert spaces, and both the queries and the messages can be any arbitrary quantum state of the input Hilbert space. In addition, the quantum oracle describing the evaluation of quantum primitives are represented more generally as unitary matrices. We describe the generalisation of the primitive to the quantum primitives as follows.\\

Let $\F = (\ES, \E, \V)$ be a quantum primitive now, with $\ES$, $\E$, and $\V$ being the setup, evaluation, and verification algorithms respectively. The game defined in the Figure~\ref{fig:game}, captures the notion of unforgeability for quantum primitives same as the classical ones, with the following modifications:
\paragraph{\bf Setup:} In the setup phase the oracles are being instantiated according to the parameters generated by $\C$, Here the evaluation oracle is defined according to Equations~\ref{qoracle-quantum-det} and~\ref{qoracle-quantum-rand} for deterministic and randomised primitives respectively and the verification oracle implements a quantum test algorithm as defined in the Definition~\ref{def:test}.

\paragraph{\bf Learning phase:}
The learning phase is similar to the classical primitives, where $\{\rho^{in}_i\}^q_{i=1}$ represent input chosen message queries and $\{\rho^{out}_i\}^q_{i=1}$ is the respective outputs after the interaction with the oracle sent to $\A$ by the challenger.

\paragraph{\bf Challenge phase:} Here the main difference is that $\M = \HilD$ is a Hilbert space and $m = \rho_{m} \in \HilD$ is a quantum challenge in the $D$-dimensional Hilbert space. In the $\qUni$ challenge phase where the message is chosen by the challenger $\C$ uniformly at random from the set of all the messages, for quantum primitives it should be selected uniformly according to the Haar measure from $\HilD$. We also need to mention that for $\qSel$ challenge phases, $\A$ is required to submit the classical description of the quantum state $\rho_m$. This is important for the verification phase, as it allows the challenger to prepare the required number of copies of the correct output for the verification. 

\paragraph{\bf Guess phase:} In this phase, the adversary submits their forgery $t$ for challenge $m$, which are now both quantum states. They win the game if $t$ passes the verification algorithm with high probability. The distinguishability condition on the message with the learning phase queries needs to be satisfied exactly as was the case for classical primitives. Here it can be seen that this is the most natural way of characterising the forgery for quantum primitives since the difference between quantum states is usually measured by their indistinguishably and their quantum distance measures.

The main difference in this phase is the difference between the classical and quantum verification procedure. The verification is fairly simpler for classical primitives since the equality can be easily checked while as for quantum primitives both message and forgery are quantum states and the verification oracle $\vO_U$ should call a quantum test algorithm $\T$ that checks the equality of quantum states as in the Definition~\ref{def:test}. Note that the challenger can prepare copies of correct outputs locally.

With the above considerations, one can use the same security game and the definitions of existential, selective and universal unforgeability as defined in Section~\ref{sec:game}. Here we only need to discuss the notion of overlap probability for quantum primitives separately due to the generality of the quantum oracles.

\subsubsection{Overlap probability in \sufm\ Definition for quantum primitives}
For quantum primitives, it is clear that the adversary's success probability in finding the output by measurement strategy is almost zero and hence defining the $P_{ov}$ as defined by Definition~\ref{def:pov-standard-orc} leads to zero overlap probability. However, in this case, as well, there is another scenario that may lead to unavoidable attacks, which is due to the error produced by the quantum test algorithm in distinguishing the states with certain overlap. An example of this is the SWAP test which has a one-sided error of $\frac{1}{2}$ even for perfectly distinguishable states. This is a fundamental difference between the quantum world and classical primitives where equality can be checked deterministically. To have a general characterisation of $P_{ov}$ for the quantum primitives, this probability needs to be defined concerning the test algorithm as follows. 

\begin{definition}[$P_{ov}$ for quantum primitives]\label{def:pov-quantum}
Let $\rho_{max}$ be the input learning phase query with the maximum overlap with the challenge state $\ket{\psi}$, allowed by the $\mu$-distinguishability condition. Let the $\eO_U$ be the unitary oracle for the quantum primitive applying $\Ue$ to the quantum inputs and let $\vO$ implement a quantum test algorithm $\T$. Then $\rho^{out}_{max} = \Ue \rho_{max} \Ue^{\dagger}$ is the output of the query from the oracle and $\rho^{out} = \ket{\psi^{out}}\bra{\psi^{out}} = \Ue \ket{\psi}\bra{\psi} \Ue^{\dagger}$ is the correct output of the challenge $\ket{\psi}$. We define the $P_{ov}$ as the error probability of the test algorithm $\T$ on distinguishing $\rho^{out}_{max}$ and $\rho^{out}$ as follows:
\begin{equation}
    P_{ov} = Pr[1 \leftarrow \T((\rho^{out}_{max})^{\otimes \kappa}, (\rho^{out})^{\otimes \kappa})]
\end{equation}
\end{definition}

This definition also implies an intuitive and practical approach to determine the desired $\mu < 1$ for quantum primitives, as it states that for any specific quantum primitive or the protocols based on that primitive, the $\mu$ should not allow for above overlap attacks with a probability larger than the required security threshold. Nevertheless, if one assumes a reasonably good quantum test algorithm, this probability for quantum primitives is usually less than the classical ones due to quantum state distinguishability and lack of adversary's knowledge over the transformation of the output bases.

\section{Results for Generalised Quantum Unforgeability of Quantum Primitives}\label{sec:results-quantum}

In this section, we present unforgeable quantum primitives for each level of our generalised unforgeability framework. Most of our positive results are based on \pru\ assumption, which is the quantum equivalent of \qprf s.

\subsection{Existential and Selective unforgeable deterministic quantum primitives }
First we show that \pru\ implies \euf\ and \suf and hence deterministic quantum primitives under this assumption can satisfy this level of unforgeability.
\begin{theorem}\label{th:suf-pru}
$\pru$ quantum primitives are \suf\ (\euf) secure.
\end{theorem}
\begin{proof}
We prove by contradiction. Let $\A$ be an adversary who wins the $\suf$ game with non-negligible probability (Note that here $P_{ov} = 0$). $\A$ selects a message $m$ before (or after) the learning phase and then outputs the respective $t$ such that it passes the verification test with non-negligible probability. Also by definition of \suf, $m \notmu \rho^{in}$ for $\mu = 1$ and hence the message $\rho_m$ is completely orthogonal to  all $\rho^{in}_i$. Now we construct an adversary $\A'$ who is playing the \pru\ game. Let $\A'$ first query all the learning phase states of $\A$ and then also issue one more query which is $\rho_{m}$. Then $\A'$ calls $\A$ and receives the input-output pair of $(m,t)$ such that $\rho_t$ is non-negligibly close to the actual output, i.e.
\begin{equation}
    F(\rho_t, \Ue\rho_m\Ue^{\dagger}) = \nonnegl(\lambda)
\end{equation}
Now $\A'$ can use this last query as a distinguisher between PRU and a unitary picked from Haar measure since $\A'$ can estimate the output with non-negligible fidelity if the $U_k$ had been picked from the family. Let $\A'$ runs a quantum equality test as described in Definition~\ref{def:test} on the $U_k \ket{\psi}$ obtained in the learning phase and $\rho_t$. Also note that if $U$ is picked from the Haar measure family, the probability of producing the output is negligible by definition. Thus whenever the test shows equality, $\A'$ can conclude that the unitary has been picked from PRU. Thus for $\A'$ we have:
\begin{equation}
    \underset{U \leftarrow U_k}{Pr}[\A'^{U}(1^{\lambda})=1] - \underset{U_{\mu} \leftarrow \mu}{Pr}[\A'^{U_{\mu}}(1^{\lambda})=1]=\nonnegl(\lambda)
\end{equation}
Which is a contradiction and the theorem has been proved. \qed
\end{proof}

\subsection{Selective unforgeable randomised quantum primitives }\label{randomised-const-quantum}
Similar to the classical constructions, for quantum primitives too, we can use randomisation to effectively secure them. The main idea is to select a new unitary transformation for each query using a classical randomness register. In this case, we need to clarify how such randomised quantum oracles can be implemented in a way that the overall transformation remains a specific unitary. 
By recalling the abstract representation of the randomised quantum oracle that we have given in the preliminary, the input state $\ket{\psi_m} = \sum_i \alpha_i\ket{m_i}$ (where $\{\ket{m_i}\}$ is a set of orthonormal bases) is mapped to a state $U(r)\ket{\psi_m} = \sum_i \beta_i(r)\ket{m_i}$ where $U(r)$ depends on the randomness and different for each query i.e. the oracle uses its internal register $\ket{r}_{\Ora}$ to activate different $U(r)$ unitaries. However, for many constructions this randomness value $r$ or a function of it like $g(r)$, will be necessary for verification and hence need to also be outputted. On the other hand, the register $\ket{r}_{\Ora}$ is the internal register of the oracle re-initiated for each query and some problems may arise if the adversary gets access to this register (see Preliminary), thus in order to be able to output this value we expand the query space and we allow the input queries to be $\ket{\mathtt{0}}\otimes\ket{\psi_m}$. We formulate the oracle as follows:
\begin{equation}
\reO_U: \ket{r}_{\Ora}\otimes \ket{\mathtt{0}}\otimes\ket{\psi_m} \rightarrow [\mathcal{I}\otimes\mathcal{I}\otimes U(r)]\ket{r}_{\Ora}\ket{r}\ket{\psi_m}
\end{equation}

Note that for the purpose of our construction, in what follows, we assume that the ancillary state is initiated as a separable state $\ket{\mathtt{0}}$ for simplicity, although if the adversary's ancillary register has not been initiated to zero, the randomness can be XORed to that value. The above oracle can be realised in different ways but we give an explicit example in the circuit model, shown in Figure~\ref{fig:qrandomised-oracle-circuit}. The input to the unitary evaluation of the oracle consists of two parts; one part includes the query and the second part is the internal randomness register which is initiated to a new value or equivalently to a new basis, for each query. This part in general acts as control qubits for the gates in the other part of the register that leads to apply a new overall unitary on the main query state. We note that the randomness register itself will remain untouched throughout the evaluation and finally its value is recorded in the $\ket{\mathtt{0}}$ part of the input query. We note that this last recording part is not in contrast with the no-cloning theorem as the $\ket{r}_{\Ora}$ is always in the computational basis.

\begin{figure}
   \centering
     \includegraphics[width=1\textwidth]{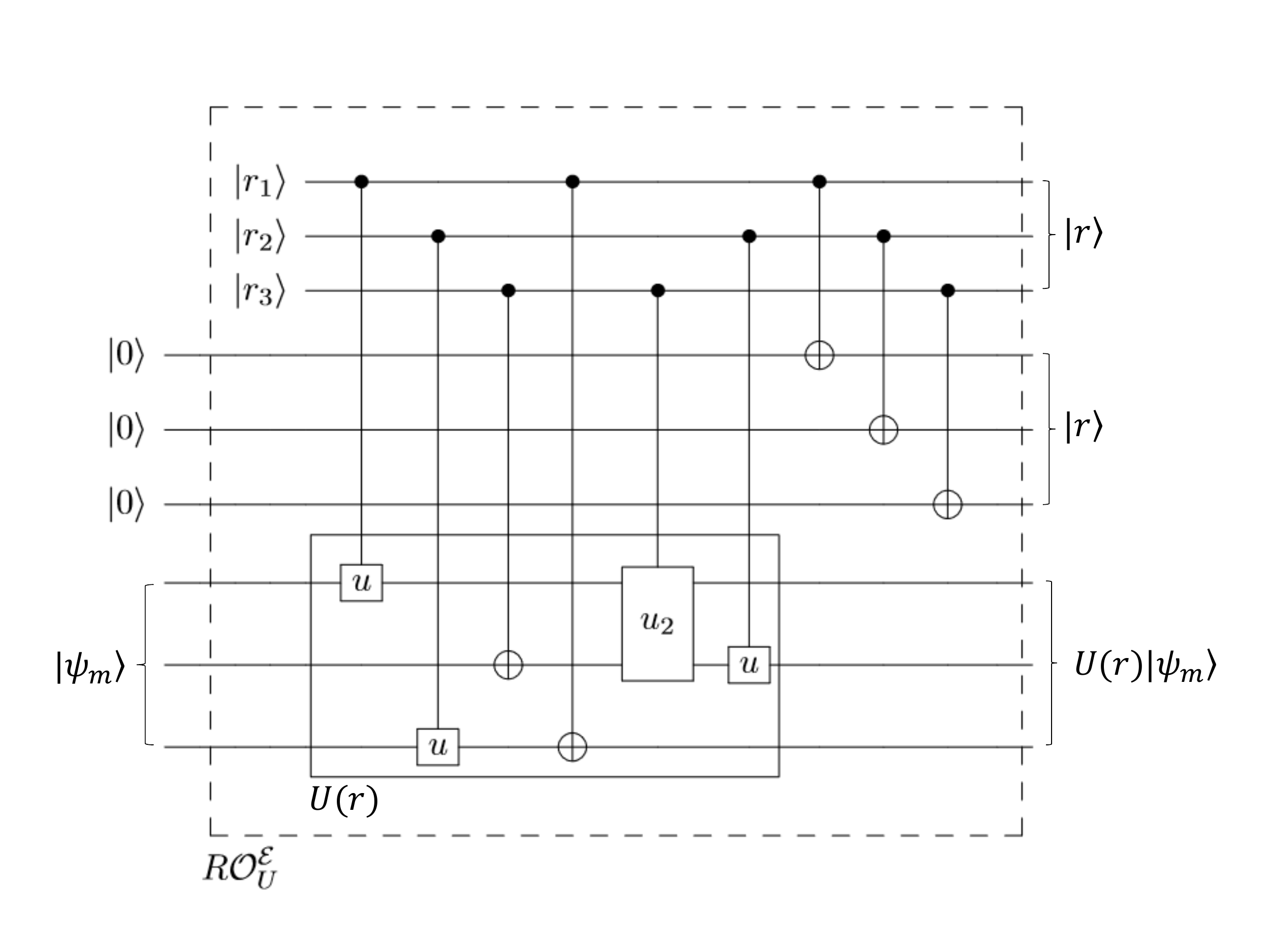}
     \caption{A sample circuit for randomised quantum oracle for quantum primitives. On each input query $\ket{\mathtt{0}}\ket{\psi_m}$, a new randomness is initialised and the random unitary $U(r)$ acts on $\ket{\psi_m}$. The random unitary $U(r)$ consists of single and 2-qubit unitary gates selected at random in the setup phase, from a gate set required to construct any unitary $U(r)$ in the family $\mathcal{U}$ specified by the construction. These single and two-qubit gates are controlled by the randomness values $\ket{r} = \ket{r_1, r_2, r_3}$. In the last step, the classical value of randomness is recorded in the ancillary qubits of the query to be returned for verification.}\label{fig:qrandomised-oracle-circuit}
 \end{figure}

Now it can be seen that in such randomised oracles, the security of the quantum primitive, lies on the assumptions on the family of $U(r)$s generated for each $r$. For instance, it is intuitive that a primitive where $U(r)$ are Haar random unitaries can be secure since the overall adversary's state after issuing polynomial queries to the oracle is almost indistinguishable from a totally mixed state. Although this assumption might be too strong. Hence we give a construction based on \pru\ which is also the quantum analogue of \qprf\ that we have used in our previous classical construction.  
 
\begin{construction}\label{const:quanum-random-pru}
Let $\PRIM=(\ES, \E, \V)$ be a quantum primitive with the evaluation unitary $\Ue: \HilR\otimes\HilD \rightarrow \HilR\otimes\HilD$ where $D$ is the overall dimension of the query and $\HilR$ is a $2^l$ dimensional Hilbert space for the randomness. And let $\lambda$ be the security parameter and $l$ and $\log(D)$ be polynomial in $\lambda$. Also, let $\mathcal{U}_{\pru} = \{U_r\}^{L}_{r=0}$ be a \pru\ family with a cardinality $L$ to be at least $2^l$. The construction is defined as follows:
\begin{itemize}
    \item{\bf Setup:} The required parameters \texttt{param} is generated to instantiate the oracles.
    \item{\bf Evaluation:} The evaluation picks randomness $r \xleftarrow{\text{\$}} \R$ uniformly, initialises the randomness register to $\ket{r}_{\Ora}$ and applies the following unitary, on each input query $\ket{\psi_m} = \sum_i \alpha_i\ket{m_i}$ where each $U(r) = U_r \in \mathcal{U}_{\pru}$
    \begin{equation}\label{eq:random-quantum-oracle-const3}
        \reO_U: \ket{r}_{\Ora}\ket{\mathtt{0}}\ket{\psi_m} \overset{U_{\mathcal{E}}}{\rightarrow} [\mathcal{I}\otimes\mathcal{I}\otimes U(r)]\ket{r}_{\Ora}\ket{r}\ket{\psi_m}
    \end{equation}
    \item{\bf Verification:} The verification oracle calls a quantum test algorithm $\T$ as defined in Definition~\ref{def:test} on $U(r)\ket{\psi_m}\bra{\psi_m}U(r)^{\dagger}$ and the tag state $\rho_t$:
    \begin{itemize}
        \item If $F(\rho_t, U(r)\ket{\psi_m}\bra{\psi_m}U(r)^{\dagger}) = 1 - \negl(\lambda)$ return $\top$ with a probability $1 - \negl(\lambda)$
        \item and $Pr[1 \leftarrow \T[(\Ue\rho_{\delta}\Ue^{\dagger})^{\otimes \kappa_1}, (\Ue \rho_m \Ue^{\dagger})^{\otimes \kappa_2}]] = \negl(\lambda)$ for any state $\rho_{\delta}$ with $\delta^2$-indistinguishable from $\rho_m$.
    \end{itemize}
\end{itemize}
\end{construction}

\begin{theorem}
Construction~\ref{const:quanum-random-pru} is \sufm\ secure for any $\mu \geq 1-\delta^2$.
\end{theorem}
\begin{proof}
We prove by contradiction. Let $\A$ be a QPT adversary who plays the \sufm\ game where the evaluation oracle is as shown in the equation(~\ref{eq:random-quantum-oracle-const3}), and wins with non-negligible probability in the security parameter \emph{i.e.} $\A$, wins the game by producing a valid tag $\rho_t$ for their selected message $\rho_m$ and randomness $r^*$ with the following probability, after interacting with the oracle in the learning phase:
\begin{equation}
    Pr[1\leftarrow \GCM{\F}{\qSel, \mu}(\lambda, \A)] - P_{ov} = \nonnegl(\lambda)
\end{equation}
Where the $P_{ov} = Pr[1 \leftarrow \T(\rho^{out}_{max})^{\otimes \kappa_1}, (\Ue \rho_m \Ue^{\dagger})^{\otimes \kappa_2}]$ according to Definition~\ref{def:pov-quantum}, and $\rho^{out}_{max}$ is query with maximum allowed overlap from $\mu$-distinguishability condition. Since the construction implies that $P_{ov} = \negl(\lambda)$, this means:
\begin{equation}\label{eq:const3-adversary-wins-probability}
    Pr[1\leftarrow \GCM{\F}{\qSel, \mu}(\lambda, \A)] = \nonnegl(\lambda)
\end{equation}
Consequently, $\A$ can produce an output $\rho_t$ with non-negligible fidelity with the actual output $\U(r^*)\rho_m\U(r^*)^{\dagger}$, for a $\U_{r^*} \in \mathcal{U}_{\pru}$.
Now we consider two cases. Either $r^*$ is one of the randomnesses that $\A$ has received during the learning phase, which means $\A$ can closely approximate the output of a random unitary $U(r^*)$ from a single query, or $r^*$ is a new randomness value, for a new random unitary $U(r^*)$ where $\A$ has no query on it. We will show that each case leads to a contradiction.

First, we show that $\A$'s output state after the learning phase, i.e. $\sigma_{out}$ cannot include more than a one-dimensional subspace of each of the $U(r)$ unitaries. To cover a subspace with a dimension of at least two, $\A$ needs to find a common output basis from two different queries. On the other hand, we note that as shown in~\cite{song2017quantum}, any PRUs are generators of Pseudorandom Quantum States (PRS) that are a family of quantum states computationally indistinguishable from Haar measure. Hence the joint output states $\sigma_{out}$ is also indistinguishable from Haar random states for $\A$ who is a QPT adversary. Now if $\A$ can find a common output subspace, it means that there are at least two states, corresponding to the bases of the 2-dimensional subspace, that are indistinguishable (or $0$-distinguishable according to Definition~\ref{def:dist}), and hence $\A$ can use those queries to distinguish the distribution of states $\sigma_{out}$ and a Haar random distribution which contradicts the fact that the oracle will generate a PRS set of states after $q$ queries.
Now we show that each case will lead to a contradiction. We start with the second case where if $\A$ produces an indistinguishable (concerning $\T$) output for a random unitary with no query, then $\A$ can perform the learning phase locally without any interaction with the oracle and hence produce the output of any unitary picked from a family indistinguishable to Haar measure, which is a clear contradiction.
For the first case, relying on the previous argument, we rewrite the learning phase states of the $\A$ after $q$ queries, as follows:
\begin{equation}
    \sigma_{in} = \ket{\phi_{r^*}}\bra{\phi_{r^*}} \otimes \sigma^{q-1}_{in}, \quad \sigma_{out} = U_{r^*}\ket{\phi_{r^*}}\bra{\phi_{r^*}}U_{r^*}^{\dagger} \otimes \sigma^{q-1}_{out}
\end{equation}
where $\ket{\phi_{r^*}}$ is the query associated to $U_{r^*}$ for which $\A$ produces a forgery and $\sigma^{q-1}_{in}$ and $\sigma^{q-1}_{out}$ are the input and output states of the remaining $q-1$ query respectively. We note that $\sigma^{q-1}_{out}$ consists of $q-1$ quantum states with a distribution $\delta$ over a $D'$-dimensional Hilbert space s.t. $\delta$ is Haar-indistinguishable. Furthermore, the ancillary register where the $r$ is encoded consists of $q$ independent random values. Now let us construct an adversary $\A'$ who is a \pru\ distinguisher. Let $\A'$ interact with a unitary $U$ either selected from $\mathcal{U}_{\pru}$ or from Haar measure, and query a state $\ket{\phi_{r^*}}$ as described above, and returns $U\ket{\phi_{r^*}}$ together with an ancillary register $\ket{r}$ where $r$ picked uniformly at random. Then $\A'$ also locally creates $q-1$ Haar-random states and returns to $\A$ as the $ \sigma^{q-1}_{out}$. Then $\A'$ also queries $\rho_m$ from the oracle. Now $\A'$ uses the same test algorithm $\T$ to check the output of $\A$ i.e. $\rho_t$ with the the oracle's output for the last query which is $U\rho_m U^{\dagger}$. From equation~(\ref{eq:const3-adversary-wins-probability}), we know that this probability is non-negligible, while as for a Haar random unitary the probability is negligible, thus can conclude that 
\begin{equation}
    |\underset{r \leftarrow \R}{Pr}[\A'^{U_r}(1^{\lambda})=1] - \underset{U \leftarrow Haar}{Pr}[\A'^U(1^{\lambda})=1]| = \nonnegl(\lambda).
\end{equation}
which is a contradiction and the theorem has been proved. \qed
\end{proof}

\subsection{Universal unforgeable quantum primitives}
It has been previously shown in~\cite{arapinis2019quantum} that certain quantum primitives, like quantum PUFs where their evaluation function satisfies \uu\ condition, can be secure \emph{wrt.} this level of unforgeability\footnote{There as the unforgeability has been studied in the context of PUFs this level of unforgeability is called \emph{selective unforgeability} while as here we call it universal unforgeability.}. Here we generalise this result for general quantum primitives to the \pru\ assumption.

Now we establish our general positive result for quantum primitives.

\begin{theorem}
Deterministic quantum \pru\ and \uu\ schemes are \uuf\ secure.  
\end{theorem}
\begin{proof}
We prove this implication independent of the results of~\cite{arapinis2019quantum},  from our previously established results. From Theorem~\ref{th:qprf-1euf} we know that \pru\ primitives are \suf\ secure. Also from Theorem~\ref{th:suf-uuf} we have shown that \uuf\ is weaker than \suf. Thus any \pru\ primitive is \uuf\ secure.\qed
\end{proof}

\section{Conclusion and future directions}
\label{sec:conclusion-disc}
We have presented new fine-grained definitions of quantum unforgeability that unify different levels of unforgeability, different types of primitives, and better capture the properties of quantum adversaries. In particular, the parameterised definitions for selective and existential unforgeability lead to some non-trivial no-go results. More precisely, our Theorem~\ref{th:sel-qCM} shows that non-randomised MAC schemes such as HMAC and NMAC cannot satisfy existential and selective unforgeability except for $\mu = 1$ and hence are always vulnerable against more powerful quantum adversaries. On the other hand, our randomised construction shows a fix to this problem and presents an approach towards proper randomization of classical primitives such that they can resist emulation type of attacks. Furthermore, we have shown that a similar technique can be applied to quantum primitives to construct randomised \sufm\ secure schemes (\ref{randomised-const-quantum}). Nevertheless, constructing efficient randomised oracles for quantum primitives using random quantum circuits or t-designs is an interesting future research direction. We have also shown that universal unforgeability is a level of security that both deterministic quantum and classical primitives can achieve. Although this is a weaker definition, it is enough for many practical purposes where unforgeability is the desired property, such as identification. Finally, it would be interesting to see the applicability of our definition and framework in practice to specific quantum primitives such as quantum money and classical public-key primitives such as digital signatures, which we also leave as a future research direction. A summary of all the possibility and impossibility results in this paper have been given in Table~\ref{table:summary}.

\begin{table}
\centering
\captionsetup{font=small}
\resizebox{\textwidth}{!}{
\begin{tabular}{|c|c|c|c|c|c|}
\hline
\backslashbox{Primitives}{qGU.level} & \euf & \eufm $(\mu \neq 1)$ & \suf & \sufm $(\mu \neq 1)$ & \uuf \\ \hline
\multirow{2}{*}{Classical}& \multirow{2}{*}{\qprf} & \multirow{2}{*}{$\times$} & \multirow{2}{*}{\qprf} & det: $\times$ & \multirow{2}{*}{\qprf}\\\cline{5-5}
    & & & & rand: Construction~\ref{const:classical-random-const-prf} & \\
\hline
\multirow{2}{*}{Quantum}& \multirow{2}{*}{\pru} & \multirow{2}{*}{$\times$} & \multirow{2}{*}{\pru} & det:$\times$ & \multirow{2}{*}{\pru, \uu}\\\cline{5-5}
    & & & & rand: Construction~\ref{const:quanum-random-pru} & \\
\hline
\end{tabular}}
\caption{Summary of the possibility and impossibility results in the quantum Generalised Unforgeability definition for classical and quantum primitives. \qprf\ and \pru\ refer to non-randomised primitives with an evaluation selected from such families, and $\times$ denotes that there are no primitives secure in that level of unforgeability.}
\label{table:summary}
\end{table}



\bibliographystyle{ieeetr}
\bibliography{qunforgeability.bib}

\begin{thebibliography}{10}

\bibitem{boneh2013quantum}
D.~Boneh and M.~Zhandry, ``Quantum-secure message authentication codes,'' in
  {\em Advances in Cryptology -- EUROCRYPT 2013} (T.~Johansson and P.~Q.
  Nguyen, eds.), (Berlin, Heidelberg), pp.~592--608, Springer Berlin
  Heidelberg, 2013.

\bibitem{boneh2013secure}
D.~Boneh and M.~Zhandry, ``Secure signatures and chosen ciphertext security in
  a quantum computing world,'' in {\em Advances in Cryptology -- CRYPTO 2013}
  (R.~Canetti and J.~A. Garay, eds.), (Berlin, Heidelberg), pp.~361--379,
  Springer Berlin Heidelberg, 2013.

\bibitem{kaplan2016breaking}
M.~Kaplan, G.~Leurent, A.~Leverrier, and M.~Naya-Plasencia, ``Breaking
  symmetric cryptosystems using quantum period finding,'' in {\em Advances in
  Cryptology -- CRYPTO 2016} (M.~Robshaw and J.~Katz, eds.), (Berlin,
  Heidelberg), pp.~207--237, Springer Berlin Heidelberg, 2016.

\bibitem{gagliardoni2016semantic}
T.~Gagliardoni, A.~H{\"u}lsing, and C.~Schaffner, ``Semantic security and
  indistinguishability in the quantum world,'' in {\em Advances in Cryptology
  -- CRYPTO 2016} (M.~Robshaw and J.~Katz, eds.), (Berlin, Heidelberg),
  pp.~60--89, Springer Berlin Heidelberg, 2016.

\bibitem{eurocrypt-2020-30239}
G.~Alagic, C.~Majenz, A.~Russell, and F.~Song, ``Quantum-access-secure message
  authentication via blind-unforgeability,'' in {\em 39th Annual International
  Conference on the Theory and Applications of Cryptographic Techniques,
  Zagreb, Croatia, May 10–14, 2020, Proceedings}, vol.~12105 of {\em Lecture
  Notes in Computer Science}, Springer, 2020.

\bibitem{kaplan2015quantum}
M.~Kaplan, G.~Leurent, A.~Leverrier, and M.~Naya-Plasencia, ``Quantum
  differential and linear cryptanalysis,'' {\em arXiv preprint
  arXiv:1510.05836}, 2015.

\bibitem{santoli2016using}
T.~Santoli and C.~Schaffner, ``Using simon's algorithm to attack symmetric-key
  cryptographic primitives,'' {\em arXiv preprint arXiv:1603.07856}, 2016.

\bibitem{goldwasser1988digital}
S.~Goldwasser, S.~Micali, and R.~L. Rivest, ``A digital signature scheme secure
  against adaptive chosen-message attacks,'' {\em SIAM Journal on computing},
  vol.~17, no.~2, pp.~281--308, 1988.

\bibitem{an2002security}
J.~H. An, Y.~Dodis, and T.~Rabin, ``On the security of joint signature and
  encryption,'' in {\em International Conference on the Theory and Applications
  of Cryptographic Techniques}, pp.~83--107, Springer, 2002.

\bibitem{boneh2006strongly}
D.~Boneh, E.~Shen, and B.~Waters, ``Strongly unforgeable signatures based on
  computational diffie-hellman,'' in {\em International Workshop on Public Key
  Cryptography}, pp.~229--240, Springer, 2006.

\bibitem{bellare1995xor}
M.~Bellare, R.~Gu{\'e}rin, and P.~Rogaway, ``Xor macs: New methods for message
  authentication using finite pseudorandom functions,'' in {\em Annual
  International Cryptology Conference}, pp.~15--28, Springer, 1995.

\bibitem{bellare2000security}
M.~Bellare, J.~Kilian, and P.~Rogaway, ``The security of the cipher block
  chaining message authentication code,'' {\em Journal of Computer and System
  Sciences}, vol.~61, no.~3, pp.~362--399, 2000.

\bibitem{dodis2012message}
Y.~Dodis, E.~Kiltz, K.~Pietrzak, and D.~Wichs, ``Message authentication,
  revisited,'' in {\em Annual International Conference on the Theory and
  Applications of Cryptographic Techniques}, pp.~355--374, Springer, 2012.

\bibitem{alwen2014key}
J.~Alwen, M.~Hirt, U.~Maurer, A.~Patra, and P.~Raykov, ``Key-indistinguishable
  message authentication codes,'' in {\em International Conference on Security
  and Cryptography for Networks}, pp.~476--493, Springer, 2014.

\bibitem{bellare2004power}
M.~Bellare, O.~Goldreich, and A.~Mityagin, ``The power of verification queries
  in message authentication and authenticated encryption.,'' {\em IACR Cryptol.
  ePrint Arch.}, vol.~2004, p.~309, 2004.

\bibitem{gagliardoni2017quantum}
T.~Gagliardoni, ``Quantum security of cryptographic primitives,'' {\em arXiv
  preprint arXiv:1705.02417}, 2017.

\bibitem{alagic2018unforgeable}
G.~Alagic, T.~Gagliardoni, and C.~Majenz, ``Unforgeable quantum encryption,''
  in {\em Advances in Cryptology -- EUROCRYPT 2018} (J.~B. Nielsen and
  V.~Rijmen, eds.), (Cham), pp.~489--519, Springer International Publishing,
  2018.

\bibitem{marvian2016universal}
I.~Marvian and S.~Lloyd, ``Universal quantum emulator,'' {\em arXiv preprint
  arXiv:1606.02734}, 2016.

\bibitem{arapinis2019quantum}
M.~Arapinis, M.~Delavar, M.~Doosti, and E.~Kashefi, ``Quantum physical
  unclonable functions: Possibilities and impossibilities,'' {\em arXiv
  preprint arXiv:1910.02126}, 2019.

\bibitem{doosti2020client}
M.~Doosti, N.~Kumar, M.~Delavar, and E.~Kashefi, ``Client-server identification
  protocols with quantum puf,'' {\em arXiv preprint arXiv:2006.04522}, 2020.

\bibitem{gagliardoni2020make}
T.~Gagliardoni, J.~Kr{\"a}mer, and P.~Struck, ``Make quantum
  indistinguishability great again,'' {\em arXiv preprint arXiv:2003.00578},
  2020.

\bibitem{chevalier2020security}
C.~Chevalier, E.~Ebrahimi, and Q.-H. Vu, ``On the security notions for
  encryption in a quantum world,'' tech. rep., IACR Cryptology ePrint Archive,
  2020: 237, 2020.

\bibitem{wootters1982single}
W.~K. Wootters and W.~H. Zurek, ``A single quantum cannot be cloned,'' {\em
  Nature}, vol.~299, no.~5886, p.~802, 1982.

\bibitem{buhrman2001quantum}
H.~Buhrman, R.~Cleve, J.~Watrous, and R.~De~Wolf, ``Quantum fingerprinting,''
  {\em Physical Review Letters}, vol.~87, no.~16, p.~167902, 2001.

\bibitem{chabaud2018optimal}
U.~Chabaud, E.~Diamanti, D.~Markham, E.~Kashefi, and A.~Joux, ``Optimal
  quantum-programmable projective measurement with linear optics,'' {\em
  Physical Review A}, vol.~98, no.~6, p.~062318, 2018.

\bibitem{ji2018pseudorandom}
Z.~Ji, Y.-K. Liu, and F.~Song, ``Pseudorandom quantum states,'' in {\em Annual
  International Cryptology Conference}, pp.~126--152, Springer, 2018.

\bibitem{armknecht2016towards}
F.~Armknecht, D.~Moriyama, A.-R. Sadeghi, and M.~Yung, ``Towards a unified
  security model for physically unclonable functions,'' in {\em
  Cryptographers’ Track at the RSA Conference}, pp.~271--287, Springer, 2016.

\bibitem{soukharev2016post}
V.~Soukharev, D.~Jao, and S.~Seshadri, ``Post-quantum security models for
  authenticated encryption,'' in {\em 7th International Workshop on
  Post-Quantum Cryptography}, pp.~64--78, Springer, 2016.

\bibitem{holevo1973bounds}
A.~S. Holevo, ``Bounds for the quantity of information transmitted by a quantum
  communication channel,'' {\em Problemy Peredachi Informatsii}, vol.~9, no.~3,
  pp.~3--11, 1973.

\bibitem{wiesner1983conjugate}
S.~Wiesner, ``Conjugate coding,'' {\em ACM Sigact News}, vol.~15, no.~1,
  pp.~78--88, 1983.

\bibitem{bozzio2018experimental}
M.~Bozzio, A.~Orieux, L.~T. Vidarte, I.~Zaquine, I.~Kerenidis, and E.~Diamanti,
  ``Experimental investigation of practical unforgeable quantum money,'' {\em
  npj Quantum Information}, vol.~4, no.~1, pp.~1--8, 2018.

\bibitem{kumar2019practically}
N.~Kumar, ``Practically feasible robust quantum money with classical
  verification,'' {\em Cryptography}, vol.~3, no.~4, p.~26, 2019.

\bibitem{boneh2011random}
D.~Boneh, {\"O}.~Dagdelen, M.~Fischlin, A.~Lehmann, C.~Schaffner, and
  M.~Zhandry, ``Random oracles in a quantum world,'' in {\em Advances in
  Cryptology -- ASIACRYPT 2011} (D.~H. Lee and X.~Wang, eds.), (Berlin,
  Heidelberg), pp.~41--69, Springer Berlin Heidelberg, 2011.

\bibitem{song2017quantum}
F.~Song and A.~Yun, ``Quantum security of nmac and related constructions,'' in
  {\em Advances in Cryptology -- CRYPTO 2017} (J.~Katz and H.~Shacham, eds.),
  (Cham), pp.~283--309, Springer International Publishing, 2017.

\end{thebibliography}
\newpage

\appendix
\section*{\LARGE Supplementary Materials}

\section{Weak and strong Quantum Generalised Unforgeability}\label{sec:weak-strong-unf}
We have formally defined our different instances of unforgeability definition as a quantum analogue of \emph{weak unforgeability}. However, the same definition with small modification can be applied to capture \emph{strong unforgeability}. First, we note that the difference between strong and weak unforgeability is only relevant to randomised primitives and for non-randomised primitives these definitions are equivalent. In the classical strong unforgeability, it is sufficient for the adversary to output a new pair to win the game and hence the adversary is allowed to pick one of the learning phase messages as the challenge and produce a new output with a fresh randomness. In our definition, it is sufficient to expand the $\mu$-distinguishability condition to the overall input of the oracle including the randomness i.e. adversary's challenge state $\ket{r^*}\bra{r^*}\otimes\rho_m$ needs to be $\mu$-distinguishable from all the learning phase states with their randomness registers which can be written as $\ket{r_i}\bra{r_i}\otimes\rho^{in}_i$. Once again for $\mu=1$ this will capture the same definition as is expected.

\section{Security proofs}
\subsection{Proof of Theorem~\ref{th:suf-uuf}: \sufm\ implies \uuf}\label{app:proof-suf-uuf}
\begin{proof}
In order to show this implication we will show that if a QPT adversary $\A$ can win in \uuf, then $\A$ can also win against \sufm. Although for simplicity we restrict the proof for the case of $\mu=1$ and the generalisation to any $\mu$ is straightforward from the hierarchy of the definition for different $\mu$ showed in the previous section. Also we recall that \suf\ and \euf\ are equivalent. Let $\A$ play the game $\GCM{\F}{\qUni}(\lambda, \A)$ by picking a set of learning phase state $\{\ket{\phi_i}\}^K_{i=1}$. Let the dimension of the unitary oracle $\eO$ be $D = 2^n$ and let the subspace of $\sigma_{in}$ be of dimension $d=poly(n)$. If $\A$ wins the game, then the average probability of $\A$ generating the an acceptable output for any $x \in \M$ picked uniformly at random by $\C$ is non-negligible:
\begin{equation}
    Pr[1\leftarrow \GCM{\F}{\qUni}(\lambda, \A)] = \underset{x \in \M}{Pr}[1\leftarrow \A(x)] = \nonnegl(\lambda).
\end{equation}
where $\underset{x \in \M}{Pr}[1\leftarrow \A(x)]$ denotes the success probability of the adversary wining the game for input $x$. Now to be able to translate this game to the \suf\ game, first we need to make sure that the set of states that $\A$ picks the challenge from them, satisfy the distinguishability condition for $\mu=1$ i.e. they are orthogonal to all the learning phase states. Let $\M'$ be the set of all the challenges with no overlap with any of the learning phase states $\rho^{in}_i$. Then we can rewrite the average success probability as follows:
\begin{equation}
\begin{split}
    \underset{x \in \M}{Pr}[1\leftarrow \A(x)] & = \underset{x \in \M'}{Pr}[1\leftarrow \A(x)]Pr[x \in \M'] + \underset{x \not\in \M'}{Pr}[1\leftarrow \A(x)]Pr[x \not\in \M'] \\
    & = \nonnegl(\lambda).
\end{split}
\end{equation}
since the dimension of the subspace that $\sigma_{in}$ spans is $d$ and it is polynomial with respect to the size of $\M$ then $\frac{|\M'|}{|\M|} \approx 1$. Hence $Pr[x \in \M'] \approx 1$ but $Pr[x \not\in \M'] = 1-Pr[x \in \M'] = \negl(\lambda)$. As a result the second term will be negligible and for the whole expression to become non-negligible, the following should hold:
\begin{equation}
\underset{x \in \M'}{Pr}[1\leftarrow \A(x)] = \nonnegl(\lambda).
\end{equation}
Now let $\A'$ be an adversary who wants to win the game $\GCM{\F}{\qSel, \mu}(\lambda, \A')$ by using $\A$. As $\A'$ picks the challenge of their choice, we will show that there is a strategy for $\A'$ to win the game relying on the average success probability of $\A$ being non-negligible over $\M'$. But also as $\A'$ is a QPT, we will show there exist a poly size subspace of $\M'$ in which $\A'$ will win with non-negligible probability. First we assume that $\M'$ is partitioned into $K$ different subsets (or subspaces) $S_i$ with equal size (or dimension in the quantum case) $|S_1|=\dots =|S_K|= l = poly(\lambda)$. Note that this partitioning is only for simplicity and any random partitioning of $\M'$ into the equal size subspace will be enough for our purpose. Now let $\A'$ pick one of the subsets of message space which consists of picking one of the $S_i$ with probability $\frac{1}{K}$. We want to show that if $\A'$ picks the $S_i$ at random and calls $\A$ on that $S_i$ the probability that in the picked subspace the following condition holds is non-negligible:
\begin{equation}\label{eq:suf-uuf-subset-prob}
\underset{x \in S_i}{Pr}[1\leftarrow \A(x)] = \nonnegl(\lambda)
\end{equation}
If this is the case, then by the definition of the average probability there exist at least one $x^*$ for which the $Pr[1\leftarrow \A(x^*)] = \nonnegl(\lambda)$ and hence the $\A'$ has won the game with a non-negligible probability. Thus we need to find the number of the success probability of $\A'$ picking a desirable subset. This probability is given by:
\begin{equation}
    Pr_{succ} = \frac{\#(S_i: \underset{x \in S_i}{Pr}[1\leftarrow \A(x)] = \nonnegl(\lambda))}{K} = \frac{Q}{K}
\end{equation}
where $Q$ denotes the number of subsets $S_i$ which satisfy the condition and $K = O(|\M'|)$. We then only need to show that $\frac{Q}{K}$ is non-negligible in the security parameter. For simplicity let us replace average probability of $\A$ in wining the game over $\M'$, with the expected value of wining probability of $\A$ over all the different elements of $\M'$ i.e.
\begin{equation}
\underset{x \in \M'}{Pr}[1\leftarrow \A(x)] = \nonnegl(\lambda) \Rightarrow \underset{\M'}{\mathbb{E}}[\A(x)] = \nonnegl(\lambda)
\end{equation}
Then we rewrite the expectation value in terms of all the subsets of $\M'$. As $\M' = S_1 \cup S_2 \cup \dots \cup S_K$, we have:
\begin{equation}
    \underset{\M'}{\mathbb{E}}[\A(x)] = \frac{1}{K}\sum^K_{i=1} \mathbb{E}_i = \nonnegl(\lambda)
\end{equation}
where $\mathbb{E}_i = \underset{S_i}{\mathbb{E}}[\A(x)]$. We then rearrange all the $\mathbb{E}_i$ descending such that the $Q$th term shows the last smallest $\mathbb{E}_i$ for which the condition is satisfied. Hence we have:
\begin{equation}
    \underset{\M'}{\mathbb{E}}[\A(x)] = \frac{1}{K}\sum^Q_{i=1} \mathbb{E}_i + \frac{1}{K}\sum^K_{i=Q+1} \mathbb{E}_i = \nonnegl(\lambda)
\end{equation}
The above equality holds if at least one of the two sums is non-negligible. If the first sum is non-negligible we have:
\begin{equation}
    \frac{1}{K}\sum^Q_{i=1} \mathbb{E}_i \geq \frac{Q\mathbb{E}_Q}{K}
\end{equation}
As $\mathbb{E}_i$s have been ordered and $\mathbb{E}_Q$ is the smallest one which is still non-negligible. Then we can conclude that:
\begin{equation}
    \frac{Q}{K} = \nonnegl(\lambda)
\end{equation}
which is what we wanted to show. The second case is when the first sum is negligible and the second sum needs to be non-negligible for the equality to hold. Similar to the previous case due to the descending ordering, we have: 
\begin{equation}
    \frac{1}{K}\sum^K_{i=Q+1} \mathbb{E}_i \leq \frac{(K-Q)\mathbb{E}_{Q+1}}{K}
\end{equation}
But followed by our assumption the $\mathbb{E}_{Q+1}$ is itself negligible and $0 < \frac{K-Q}{K} < 1$, thus this sum can never converge to a non-negligible function of $\lambda$. Hence we conclude that necessarily the first sum, and as a result $\frac{Q}{K}$ is non-negligible. Thus we have shown the equation~\ref{eq:suf-uuf-subset-prob}, and there exist a strategy for $\A'$ to win the game by calling $\A$. This concludes that \suf (\sufm) implies \uuf\ and the proof is complete.
\end{proof}

\subsection{Proof of Theorem~\ref{th:sel-qCM}: \sufm\  impossibility for deterministic primitives}\label{app:proof-selective-unf-nogo}
In this section we give a proof of Theorem~\ref{th:sel-qCM} with full details and probability analysis.
\begin{proof}
We show there is a QPT adversary $\A$ that wins the game with non-negligible probability. Let $\Ue$ be the unitary transformation corresponding to $\eO$. $\A$ runs the algorithm pictured in Figure~\ref{fig:qsel-attack}. To show that $\A$ wins the game we need to show the probability of producing a correct response for either $m$ by $\A$ is non-negligibly higher than $P_{ov}(q,\mu)$ as given by Theorem~\ref{def:pov-standard-orc}. After interacting with the oracle in the learning phase, $\A$ has the following states representing their queries and responses:
\begin{equation}
    \ket{\phi_1}\otimes\ket{\phi_r} \quad
    \ket{\phi^{out}_1}\otimes\ket{\phi^{out}_r}
\end{equation}

Now $\A$ can run a quantum emulation algorithm by setting the $\ket{\phi_r}$ as the reference state, and picking the target state to be $\ket{\psi} = \ket{m, 0}$. $\A$ uses and emulation algorithm with one block and relying on Theorem~\ref{th:qe-fins}, the output state of Stage 1 of the QE algorithm is:
\begin{equation}
\begin{split}
    \ket{\chi_f} =& \mbraket{\phi_r}{\psi} \ket{\phi_r}\ket{0} + \ket{\psi}\ket{1} - \mbraket{\phi_r}{\psi}\ket{\phi_r}\ket{1} -2\mbraket{\phi_1}{\psi}\ket{\phi_1}\ket{1}\\ & +2\mbraket{\phi_r}{\psi}\mbraket{\phi_r}{\phi_1}\ket{\phi_1}\ket{1}.
\end{split}
\end{equation}
Note that $\mbraket{\phi_1}{\psi} = 0$ and $|\mbraket{\psi}{\phi_r}|^2 = \gamma^2$ and $|\mbraket{\phi_1}{\phi_r}|^2 = 1 - \gamma^2$. Then according to Theorem~\ref{th:qe-fidel}, the fidelity of the emulation for both states is:
\begin{equation}
    F(\ket{\omega}\bra{\omega}, \Ue^{\dagger} \ket{\psi}\bra{\psi} \Ue) \geq \gamma^2(1+4(1-\gamma^2)^2)|
\end{equation}

\begin{boxfig}{$(\qSel, \mu)$-QEA: adversary's algorithm against game $\GCM{\F}{\qSel, \mu}(\lambda, \A)$}{fig:qsel-attack}\underline{\emph{$(\qSel, \mu)$-QEA}}
\\\\
{\bf Challenge phase:}
  \begin{itemize}
    \item pick $m$ as the challenge\footnote{The challenge state is $\ket{m, 0}$ which is one of the computational basis of $\Ue$.}
  \end{itemize}
{\bf First learning phase:}
  \begin{itemize}
    \item choose $\ket{\phi_1} = \ket{m', 0}$
    \item choose $\ket{\phi_r} = \sqrt{1-\gamma^2}\ket{m', 0} + \gamma \ket{m, 0}$\footnote{Set $\gamma$ to $\gamma_{max} = \sqrt{1-\mu}$ such that $m$ satisfies $m \notmu \sigma_{in}$.}
    \item Interact with the evaluation oracle $\eO_f$ and generate $\sigma$\footnote{We have $\sigma_{in} = \ket{\phi_1}\otimes\ket{\phi_r}$ as a known quantum state, and $\sigma_{out} = \ket{\phi^{out}_1}\otimes\ket{\phi^{out}_r}$ as an unknown quantum state where $m$ and $m'$ are classical bitstrings.}
  \end{itemize}
{\bf Guess phase:}
  \begin{itemize}
    \item run the quantum emulation algorithm:
    \item $\ket{\omega} \gets QE(m, \sigma_{in}, \sigma_{out})$\footnote{set the reference state of QE to $\ket{\phi_r}$.}
    \item measure $\ket{\omega}$ in the comp. basis and get $t$:
    \item output $(m, t)$
  \end{itemize}
\end{boxfig}

In general, $\gamma^2$ which is the overlap between the challenge state and the learning phase state can be as large as $1 - \mu$ allowed by the definition, thus we set the maximum allowed value of overlap which is $\gamma = \gamma_{max} = \sqrt{1 - \mu}$. Now we need to also determine $P_{ov}$ and to show whether the adversary can boost the success probability by a non-negligible value. Here one of the queries is orthogonal to the challenge and there is only one query ($\ket{\phi_r}$) with overlap, thus according to Theorem~\ref{def:pov-standard-orc} we have $P_{ov}(2, \mu) = 1 - \mu^2$. As a result
\begin{equation}
\begin{split}
    Pr[1\leftarrow \GCM{\F}{\qSel, \mu}(\lambda, \A)] - P_{ov} & = (1 - \mu)[1+4(1 - (1 - \mu))^2] - (1 - \mu^2)\\
    & = \mu(1-\mu)(4\mu - 1)
\end{split}
\end{equation}
Since $\frac{1}{4} + \nonnegl(\lambda) \leq \mu \leq 1 - \nonnegl(\lambda)$, then all the terms are non-negligible in the security parameter and this concludes the proof.
\end{proof}

\subsection{Proof of \sufm\ security for Construction~\ref{const:classical-random-const-prf} with \qprf}\label{app:proof-randomised-suf}
In this section we give a complementary proof for a \qprf\ based construction that does not need the computational definition of inter-function independence defined in Definition~\ref{def:computational-function-pairwise}. Instead, we establish the following lemma for truly random functions:

\begin{lemma}\label{lemma:pairwise-property-random-function}
Let $F: \X\rightarrow \Y$ be the family of all the functions with domain $\X$ and range $\Y$, where $\X=\{0,1\}^n$ and $\Y=\{0,1\}^m$. For any two functions $f$ and $g$ picked uniformly at random from $F$, the following pairwise property holds:
\begin{equation}\label{eq:random-function-pairwise}
    \forall x \in \X: \quad
    \underset{f, g \leftarrow F}{Pr}[f(x) = g(x)] = \negl(m)
\end{equation}
\end{lemma}
\begin{proof}
First we calculate the probability of selecting a random $f$ such that $f(x) = c$ where $c \in \Y$ is a specific element of the range. This probability is equal to the number of all the functions which return $c$ on input $x$ divided by number of all the functions in $F$ which is:
\begin{equation}
    \underset{f}{Pr}[f(x) = c] = \frac{(2^m - 1)^{(2^n - 1)}}{(2^m)^{2^n}} = \frac{1}{(2^m - 1)}(1 - \frac{1}{(2^m)^{2^n}}) \approx \frac{1}{(2^m - 1)}
\end{equation}
Since $g$ has also been picked uniformly and independently from $f$, the same probability holds for $g$. As a result $\underset{f}{Pr}[f(x) = c] = \underset{f}{Pr}[g(x) = c]$. Now we are interested in the probability where $f$ and $g$ simultaneously return $c$ which is:
\begin{equation}
    \underset{f, g}{Pr}[f(x) = g(x) = c] = \underset{f, g}{Pr}[f(x) = c \wedge g(x) = c] = (\underset{f}{Pr}[f(x) = c])^2 = \frac{1}{(2^m - 1)^2}
\end{equation}
Finally, we since we are not interested in any particular $c$, we get the following probability by considering all $c \in \Y$:
\begin{equation}
    \underset{f, g}{Pr}[f(x) = g(x)] =  |\Y|\times(\underset{f}{Pr}[f(x) = c])^2 = \frac{2^m}{(2^m - 1)^2} \approx \frac{1}{2^m} = \negl(m)
\end{equation}
Thus the proof is complete. \qed
\end{proof}

\begin{theorem}
Construction~\ref{const:classical-random-const-prf} where $F$ is a \qprf, is \sufm\ secure for any $\mu$.
\end{theorem}
\begin{proof}
We assume there exists a QPT adversary $\A$ who plays the \sufm\ game where the evaluation is according to Construction 1, and wins with non-negligible probability in the security parameter \emph{i.e.} $\A$ wins the game by producing a valid tag $t^*$ for their selected message $m^*$ and randomness $r^*$ with the following probability:
\begin{equation}\label{eq:adv-prob-wining-suf}
    Pr[1\leftarrow \GCM{\F}{\qSel, \mu}(\lambda, \A)] - P_{ov} = \nonnegl(\lambda)
\end{equation}
Where the verification algorithm checks if $F(k \oplus r^*, m^*) = t^*$. We introduce the following intermediate games:
\begin{itemize}
    \item \textbf{Game 1.} This game is similar to \sufm\ for Construction 1, except that $\A$ needs to produce forgery for an $r^*$ which is one of the previously received random values of $\{r_i\}^q_{i=1}$ in the learning phase.
    \item \textbf{Game 2.} This game is similar to Game 1, but the evaluation oracle picks a new $f$ for each query from truly random functions of the family $\F: \{0,1\}^n \rightarrow \{0,1\}^m$. Note that here the randomness value $r$, only identifies the function for each query and it is an independent random variable from the function itself. Then $\A$ needs to produce forgery $t = f(m^*)$ for the message $m^*$ that they have picked earlier in the challenge phase, as well as specify $r^*$ of the function (query) for which the forgery has been done.
\end{itemize}

First, it is straightforward that the probability of the adversary in winning \sufm\ for Construction 1, is at most negligibly higher than winning Game 1. Since $r_i$ in both cases have been picked independently and uniformly at random and the probability of producing a forgery for a specific function with no query is negligible. Thus for Construction 1, Game 1 and \sufm\ are indistinguishable.

Second, we show that Game 1 and Game 2 are indistinguishable. We prove this by contradiction. We show that if $\A$ has a non-negligible advantage in winning Game 1 over Game 2, then there exists also an adversary who can distinguish a \qprf\ with truly random functions. Let $\A$ be such an adversary. Now we construct adversary $\A'$ who is trying to distinguish a \qprf\ from truly random functions. First $\A'$ queries all the learning phase states of $\A$, and then as the last query, but also the challenge message $m^*$ selected by $\A$ as prescribed by Game 1 and Game 2. Thus due to the non-negligible advantage of $\A$ in producing a forgery for the case where the function is a \qprf\, $\A'$ can use the last query to distinguish between the two cases and we have:
\begin{equation}
    |\underset{k \leftarrow \K}{Pr}[\A'^{qPRF_k}(1^{\lambda})=1] -\underset{f \leftarrow \Y^{\X}}{Pr}[\A'^{f}(1^{\lambda})=1]| = \nonnegl(\lambda).    
\end{equation}
Which is a contradiction and we have shown that Game 1 and Game 2 are indistinguishable.

Now we recall the quantum random oracle for Construction 1, and the equivalent oracle for Game 2. Let $\reO_c$ be the random oracle for Construction 1 as follows:
\begin{equation}
    \reO_c: \sum_{m,y} \alpha_{m,y} \ket{r}_{\Ora}\ket{m,y} \rightarrow \sum_{m,y} \alpha_{m,y}\ket{r}_{\Ora}\ket{m, y \oplus (F(k \oplus r, m) || r)}
\end{equation}

For each query a new function has been picked from \qprf\ family of functions, but it is the same for all the messages in the superposition for that query. Now we also present the quantum oracle for Game 2, which is:
\begin{equation}
    \Ora_{g_2}: \sum_{m,y} \alpha_{m,y} \ket{r}_{\Ora}\ket{m,y} \rightarrow \sum_{m,y} \alpha_{m,y}\ket{r}_{\Ora}\ket{m, y \oplus (f_r(m) || r)}
\end{equation}

According to the first part of the proof, the oracles $\reO_c$ and $\Ora_{g_2}$ are equivalent. Now using Lemma~\ref{lemma:pairwise-property-random-function}, we show that each query to either of these two oracles, leads to at most a single query to an independent unitary. As a result, the adversary can at most span a one-dimensional subspace of each $U_{f_{r}}$ (resp. $U_{F(k\oplus r, .)}$) where the unitary acts on the space of the input queries excluding the part that records the randomness. To show this, we recall that each selected message $m$ inside a quantum query of the adversary corresponds to a computational basis of the Hilbert space $\HilD$ on which $U_{f_{r}}$ (resp. $U_{F(k\oplus r, .)}$) operates. Due to the pairwise independence property that we have shown in Lemma~\ref{lemma:pairwise-property-random-function}, each two randomly picked $U_{f_{r}}$ map a fixed set of computational basis, to two distinct set of computational basis. We have the following property:
\begin{equation}
\begin{split}
    & \forall m_i: \underset{f,g}{Pr}[f(m_i) = g(m_i)] = \negl(\lambda) \Rightarrow \\
    & \forall \ket{e^f_i}, \ket{e^g_i} \text{ where: } \\
    & \ket{e^f_i} = U_{f}\ket{m_i, z} = \ket{z \oplus f(m_i)}, \ket{e^g_i} = U_{g}\ket{m_i, z} = \ket{z \oplus g(m_i)} \Rightarrow\\ & \underset{f,g}{Pr}[\mbraket{e^f_i}{e^g_i} \neq 0] = \negl(\lambda)
\end{split}
\end{equation}
Which means that for any randomly picked function $f$ and $g$, the output set of the basis of the unitary, $\{e^f_i\}$ and $\{e^g_i\}$ are fully distinguishable sets of computational basis. Also since unitaries are distance preserving operators, this property holds for any sets of basis, not necessarily the computational basis. The above property holds for any two randomly picked functions of the family, \emph{i.e.} for every two queries and for any subset of the output basis including two bases which covers a 2-dimensional subspace of $\HilD$. Thus by selecting a uniformly random function for each query, we have shown that no more than a one-dimensional subspace can be spanned for that specific unitary. As the two oracles are equivalent the same thing holds for when the adversary interacts with $\reO_c$.

The rest of the proof is exactly same as the proof of Theorem~\ref{th:classical-random-const-prf-secure}, where we show that with one query to each unitary that satisfies the $\mu$-distinguishability condition with the quantum encoding of $m^*$, the success probability of $\A$ is bounded as:
\begin{equation}
    Pr[1\leftarrow \GCM{\F}{\qSel, \mu}(\lambda, \A)] \leq 1 - \mu
\end{equation}
Which is a contradiction with the assumption that $\A$ breaks the \sufm\ and the proof is complete. \qed
\end{proof}

\section{No-go result for \uuf\ security of quantum primitive against adaptive adversaries}\label{ap:uni-adaptive}
Another attack model that can be defined against \uuf\ is when we allow the adversary to use the second learning phase described in the formal definition of game in Figure~\ref{fig:game}. This attack model is stronger than the usual chosen-message attack considered for universal unforgeability and is particularly interesting for quantum primitives. This is because for a quantum primitive, the adversary receives an unknown quantum state from the challenger and enabling the second learning phase does not lead to a trivial attack. We call this attack model, adaptive-universal attack \emph(aua). Although we show that a quantum adversary who can use entanglement can break the \uuf\ security of any deterministic primitive if the second learning phase is allowed. We show this specific instance of the game as $\GCM{\F}{\qUni-aua, \mu}(\lambda, \A)$ and we note that again this instance should be parameterised with $\mu$ since a trivial attack can happen if $\A$ tries to query the challenge phase again in the second learning phase. We present our attack and general no-go result in the following theorem.

\begin{theorem}[No quantum non-randomised primitive $\F$ is aua-\uuf\ secure]\label{th:uni-aua} For any quantum primitive $\F$ and for any $\mu$ such that $0 \leq \mu \leq 1-\nonnegl(\lambda))$, there exists a QPT adversary $\A$ such that
\begin{equation}
Pr[1\leftarrow \GCM{\F}{\qUni-aua, \mu}(\lambda, \A)] = \nonnegl(\lambda).
\end{equation}
\end{theorem}

\begin{proof}
Let $\A$ be the QPT adversary playing the game $\GCM{\F}{\uuf-aua, \mu}(\lambda, \A)$ and running the algorithm described in Figure~\ref{fig:adaptive-attack}.

\begin{boxfig}{aua attack on \uuf: adversary's algorithm against game $\GCM{\F}{\qUni-aua, \mu}(\lambda, \A)$}{fig:uni-adapt}\underline{\emph{$\qUni-aua$}}
\\\\
{\bf First learning phase:} $\nul$\\

{\bf Challenge phase:}
\vspace{-0.4cm}
  \begin{enumerate}[label=]
    \item prepare qubit $\ket{0}_a$
    \item receive $\ket{\psi_m}$ as a challenge
  \end{enumerate}
{\bf Second learning phase:}
\vspace{-0.4cm}
\begin{multicols}{2}
    \begin{enumerate}[label=]
    \item $\ket{\Psi}_{ca} = CNOT_{c,a}(\ket{\psi_m}\ket{0})$  
    \item query register $c$
    \item receive $\Ue\rho_c\Ue^{\dagger}$ or $(\Ue\otimes\mathcal{I})\ket{\Psi}_{ca}$
    \item
    \item
    \item
    \item
    \item The subscript $c$ denotes the challenge and the subscript $a$ denotes the adversary's qubit.
    \item $\A$ sends the challenge part of the entangled system, $\rho_c$ as a query.
  \end{enumerate}
\end{multicols}
{\bf Guess phase:}
\vspace{-0.4cm}
\begin{multicols}{2}
  \begin{enumerate}[label=]
    \item $\ket{\psi_m^{out}}\otimes\ket{\pm} \leftarrow Measure(\ket{\Psi}_{ca}, \{\ket{\pm}\}$)
    \item \textbf{if} {$\ket{\pm} = \ket{+}$} {\setlength\itemindent{25pt} \item \textbf{output:} $\ket{t} = \ket{\psi_m^{out}}$}
    \item \textbf{else} {\setlength\itemindent{25pt} \item \textbf{output:} $\ket{t} = CZ^{\otimes n-1}(\ket{\psi_m^{out}})$}
    \item $Measure(\ket{\Psi}_{ca}, \{\ket{\pm}\}$ outputs the result of the measurement.
    \item
    \item
  \end{enumerate}
\end{multicols}
\end{boxfig}\label{fig:adaptive-attack}

$\A$ does not issue any query during the first learning phase. Then $\A$ receives an unknown challenge state $\ket{\psi_m} = \sum_{i=1}^{D} \alpha_i \ket{b_i}$ where $\{\ket{b_i}\}_{i=1}^{D}$ is a set of complete orthonormal bases for $\HilD$. Now, $\A$ prepares state $\ket{0}$ and performs a CNOT gate on the first qubit of the unknown challenge state and the ancillary qubit ($\ket{0}$) with the control qubit on the challenge state. We can assume the order of the bases is such that in the first half, the first qubit is $\ket{0}$ and in the second half the first qubit is $\ket{1}$. Then the output entangled state is
\begin{equation*}
\ket{\Psi}_{ca} = \sum_{i=1}^{D/2} \alpha_i \ket{b_i}_{c}\otimes\ket{0}_{a} + \sum_{i=\frac{D}{2} + 1}^{D} \alpha_i \ket{b_i}_{c}\otimes\ket{1}_{a}
\end{equation*}
Now we can compute the final state of the two systems after the second learning phase which is:
\begin{equation*}
\ket{\Psi^{out}}_{ca} = \sum_{i=1}^{D/2} \alpha_i (\Ue\otimes\mathbb{I})(\ket{b_i}_{c}\otimes\ket{0}_{a}) + \sum_{i=\frac{D}{2} + 1}^{D} \alpha_i (\Ue\otimes\mathbb{I})(\ket{b_i}_{c}\otimes\ket{1}_{a}).
\end{equation*}
By rewriting the first qubit in the $\ket{+}$ basis we have
\begin{equation*}
\ket{\psi_m^{out}} = [\Ue(\sum_{i=1}^{D} \alpha_i \ket{b_i}_{c})]\frac{\ket{+}}{\sqrt{2}} + [\Ue(\sum_{i=1}^{D/2} \alpha_i \ket{b_i}_{c} - \sum_{i=\frac{D}{2} + 1}^{D} \alpha_i \ket{b_i}_{c})]\frac{\ket{-}}{\sqrt{2}}.
\end{equation*}
Then, the adversary measures his local qubit in the $\{\ket{+}, \ket{-}\}$ bases. If he obtains $\ket{+}$, the state collapses to $\Ue(\sum_{i=1}^{D} \alpha_i \ket{b_i}_{c}) = \Ue\ket{\psi_m}$ that is the desired state with fidelity 1. If the output of the measurement is $\ket{-}$, half of the terms have a minus sign. In this case, $\A$ applies a controlled-Z gate on the second half of the state to obtain again $\Ue\ket{\psi_m}$. As a result, for any $\kappa_1$ and $\kappa_2$, we have:
\begin{equation*}
    Pr[1\leftarrow \GCM{\F}{\qUni-aua, \mu}(\lambda, \A)] = Pr[1\leftarrow\T((\Ue\ket{\psi_m})^{\otimes\kappa_1}, \ket{t}^{\otimes\kappa_2})] = 1.
\end{equation*}
Now to complete the proof, we show that the $\mu$-distinguishability is satisfied on average.
We need to calculate the reduced density matrix of this state and compare it with the density matrix $\rho_{\psi}=\ket{\psi}\bra{\psi}$ in terms of the Uhlmann's fidelity. The reduced density matrix of the challenge state can be calculated as follows:
\begin{equation*}
\begin{split}
\rho_c = Tr_{a}[\ket{\psi}\bra{\psi}_{ca}] = & \sum^{D}_{i=1}|\alpha_i|^2\ket{b_i}\bra{b_i} + \sum^{\frac{D}{2}}_{i=j=1}\sum^D_{j\neq i,j=\frac{D}{2}+1}\overline{\alpha_i}\alpha_j\ket{b_i}\bra{b_j} + \\ & \sum^{D}_{i=\frac{D}{2}+1}\sum^{\frac{D}{2}}_{j\neq i, j=1}\overline{\alpha_i}\alpha_j\ket{b_i}\bra{b_j}
\end{split}
\end{equation*}
where $Tr_{a}$ denoted the partial trace taken over the adversary's sub-system. And the first sum shows the diagonal terms of the density matrix. As it can be seen these density matrices are different in half of the non-diagonal terms with the $\rho_{\psi}$. According to the Uhlmann's fidelity definition in the preliminary, and the fact that $\ket{\psi}$ is a pure state the fidelity reduce to:
\begin{equation*}
F(\rho_{\psi},\rho_c)=[Tr(\sqrt{{\sqrt{\rho_{\psi}}}\rho_c {\sqrt{\rho_{\psi}}}})]^{2} = \bra{\psi}\rho_c\ket{\psi} = \sum^D_{i=1}|\alpha_i|^2\bra{b_i}\rho_c\ket{b_i}.
\end{equation*}
By substituting the $\rho_c$ from above, the result will be as follows:
\begin{equation*}
F(\rho_{\psi},\rho_c)= \sum^D_{i=1}|\alpha_i|^4 + \sum^{\frac{D}{2}}_{i=1}\sum^{D}_{j=\frac{D}{2}+1}2|\alpha_i\alpha_j|^2 = 1 - \sum^{\frac{D(D-1)}{4}}_{i=1}2|\gamma_i|^2
\end{equation*}
where $|\gamma_i|^2$ denoted the square of a quarter of the non-diagonal elements of $\rho_{\psi}$. This is a positive value and on average over all the state $\ket{\psi}$, non-negligible compared to the dimensionality of the state. Hence:
\[F(\rho_{\psi},\rho_c) \leq 1 - \nonnegl(\lambda)\]
and the distinguishability condition is satisfied and the proof is complete. \qed
\end{proof}



\end{document}